\documentclass[journal,10pt,draftcls,onecolumn]{IEEEtran}

\usepackage{graphics} 
\usepackage{epsfig} 
\usepackage{mathptmx} 
\usepackage{times} 
\usepackage{amsmath} 
\usepackage{amssymb}  
\usepackage{enumerate}
\usepackage{mathrsfs}
\usepackage{multirow}
\usepackage{subfigure}
\usepackage{latexsym}
\usepackage{bm}
\usepackage{dsfont}
\usepackage{graphicx}
\usepackage{algpseudocode}
\usepackage{algorithm}   

\newtheorem{remark}{Remark}
\newtheorem{lemma}{Lemma}
\newtheorem{proof}{Proof}
\newtheorem{proposition}{Proposition}

\newtheorem{corollary}{Corollary}
\newtheorem{theorem}{Theorem}

\newtheorem{assumption}{Assumption}

\title{Partition-based Distributed Kalman Filter with plug and play features}
\author{Marcello Farina\thanks{Dipartimento di Elettronica, Informazione e Bioingegneria, Politecnico di Milano, via Ponzio 34/5, 20133 Milan, Italy}, Ruggero Carli\thanks{Department of Information Engineering University of Padova, Via Gradenigo 6/B, 35131, Padova, Italy} %
}

\begin{document}

\maketitle

\begin{abstract}                          
In this paper we propose a novel partition-based distributed state estimation scheme for non-overlapping subsystems based on Kalman filter. The estimation scheme is designed in order to account, in a rigorous fashion, for dynamic coupling terms between subsystems, and for the uncertainty related to the state estimates performed by the neighboring subsystems. The online implementation of the proposed estimation scheme is scalable, since it involves (\emph{i}) small-scale matrix operations to be carried out by the estimator embedded in each subsystem and (\emph{ii}) neighbor-to-neighbor transmission of a limited amount of data. We provide theoretical conditions ensuring the estimation convergence. Reconfigurability of the proposed estimation scheme is allowed in case of plug and play operations. Simulation tests are provided to illustrate the effectiveness of the proposed algorithm.
\end{abstract}

\section{Introduction}
In many different engineering areas there has been, in the last years, a huge effort to develop algorithms and protocols allowing a number of interconnected, possibly spatially distributed systems, devices, sensors, and actuators, to operate cooperatively and to possess self-organization capabilities. Notable examples include smart grids \cite{Resende11}, environmental monitoring systems \cite{EMMON11}, large-scale irrigation and hydraulic networks \cite{Cantoni07,MaeDoa:12-007}, and multi-robot/vehicle systems \cite{MartinezCortesBullo07,Murray07}.\\
Related research on systems of systems \cite{Samad2011} or cyber-physical systems \cite{Antsaklis2013} is nowadays fostered, pursuing several challenges, including the design of hierarchical and distributed monitoring and control systems with reliability and robustness properties with respect to uncertainties, changing environment, communication failures, etc.\\
In particular, theoretically sound distributed monitoring and state estimation methods are necessary to allow for optimal managing of sensor networks. As also discussed in the survey paper~\cite{Sijs_Lazar_et_al_2008}, two main classes of estimation techniques for distributed smart sensing schemes are presently under investigation. They are generally both referred, in the literature, to as \emph{distributed state-estimation} algorithms. While a widely-considered problem concerns the case where the full state of the system is estimated by all subsystems, e.g., based on consensus and diffusion strategies, e.g., \cite{Saber07CDC,diffusion_d_est,TACFFS09}, in this paper we focus on \emph{partition-based estimation}. The latter consists of estimating, for each sensor, only a part of the state vector of a system, using information transmitted by other neighboring sensors. This problem gives rise to low-order estimation problems solved in a distributed way, and is particularly useful when the observed systems are large scale ones, e.g., power networks~\cite{Siljac78,FP-RC-FB:12}, transport networks~\cite{s1978}, process plants~\cite{Vadigepalli2003}, and robot fleets~\cite{Mutambara_Durrant-Whyte2000}.\\
Concerning linear discrete-time systems, recent contributions include \cite{Vadigepalli2003,Khan2008,Stankovic09,Farina2010,Farina2011b,Negenborn-Kalman13,Haber13,Riverso2013b,Riverso2013e}.\\ Among these, \cite{Vadigepalli2003,Khan2008,Stankovic09,Negenborn-Kalman13} propose Kalman filter-based estimation schemes suitable for systems affected by stochastic noise. The papers \cite{Vadigepalli2003,Stankovic09} propose methods based on local Kalman prediction equations (and neglect the dynamic interconnection terms) and on consensus steps to account for possible overlapping states between pairs of subsystems. The paper \cite{Khan2008} proposes a two-step Kalman filter, where the correction step is performed by each subsystem based on local measurements, while the prediction step is based on
approximating the centralized error process using a distributed iterate-collapse inversion algorithm for $L$-banded matrices \cite{KhanMouraDICI2008}. As in \cite{Vadigepalli2003,Stankovic09}, a consensus step is used to optimally account for overlapping states. Finally, in \cite{Negenborn-Kalman13}, a prediction/corrector-based method for multi-rate systems is proposed. It is worth noting that sufficient convergence conditions are provided just in \cite{Stankovic09} which, in case of non-overlapping subsystems, basically amount to the stability of the original system.\\
The papers \cite{Farina2010,Farina2011b,Riverso2013b,Riverso2013e} assume that the system is affected by bounded noise and guarantee, under suitable conditions, convergence of the estimator and the fulfillment of constraints on local states, e.g., in \cite{Farina2010}, or estimation errors, e.g., in \cite{Farina2011b,Riverso2013b,Riverso2013e}. Finally, \cite{Haber13} proposes an approximated distributed filter based on the moving horizon estimator studied in \cite{ABB03}. A different - cooperative and iterative - approach based on Lagrange decomposition is proposed in \cite{Georges20141451}, where continuous-time systems are considered.\\
The conditions required for convergence of the estimators discussed in all the mentioned papers, where available, (with the notable exception of \cite{Riverso2013e}) require a centralized synthesis/analysis phase which (\emph{i}) limits the application to very large-scale systems and (\emph{ii}) requires a complete re-design in case of configuration changes (e.g., addition/removal of subsystems or sensors). On the other hand, in \cite{Riverso2013e} the design phase (guaranteeing global properties) is distributed, i.e., the state estimator embedded in each subsystem is devoted to solve a local design problem. This has paved the way to a plug-and-play (PnP) implementation \cite{StoustrupEJC,RiversoFarinaGFT_PnP13}, which confers flexibility, reconfigurability, and reliability to the estimation architecture.\\
In this paper we propose a novel partition-based distributed state estimation scheme based on Kalman filter (denoted DKF) for non-overlapping subsystems affected by stochastic noise. The estimation scheme is designed to account for dynamic coupling terms between subsystems, and for the uncertainty related to the state estimates performed by the neighboring subsystems. This is done in a conservative but rigorous way by means of suitable covariance matrix bounds. The online implementation of the proposed estimation scheme is scalable, since it involves (\emph{i}) small-scale matrix operations to be carried out by the estimator embedded in each subsystem and (\emph{ii}) neighbor-to-neighbor transmission of a limited amount of data. Concerning the design/analysis phase, we provide both centralized (both with suitable linear matrix inequalities and with aggregate small gain-type arguments) and distributed scalable conditions to be verified ensuring the estimation convergence. The latter are then used to provide a fully distributed and PnP implementation of DKF. More specifically, distributed reconfigurability conditions are provided in case a subsystem is added to or removed from the network, and also in case PnP operations involve sensors.

The paper is structured as follows. In Section \ref{sec:statement} we introduce and motivate the distributed Kalman filter equations, while in Section \ref{sec:convergence} we provide the main conditions for convergence. In Section \ref{sec:d-design} we discuss how DKF can be designed in a distributed fashion and the resulting application for PnP operations. Finally, in Section \ref{sec:Exs} the algorithm is tested both on an academic example and on a benchmark case study, and in Section \ref{sec:conclusions} some conclusions are drawn. All proofs are postponed to Appendix A for better readability.\\
\noindent
\textbf{Notation}\\
The symbols $\geq$ and $>$ are used to denote semi-definite positive matrices and definite positive matrices, respectively. The symbols $\mathcal{R}$ and $\mathcal{L}$ are used for brevity to denote the Riccati equation update and the optimal Kalman predictor gain, respectively, i.e.,
$$\begin{array}{lcl}
\mathcal{R}(P,A,C,Q,R)&=&A PA^T-APC^T(CPC^T+R)^{-1}{C} PA^T + Q\\
\mathcal{L}(P,A,C,R)&=&A P C^T(C P C^T+R)^{-1}
\end{array}$$
where $P$, $A$, $C$, $Q$, and $R$ are matrices of appropriate dimensions. Finally, the cardinality of a set $\mathcal{N}$ is denoted with~$|\mathcal{N}|$ and the spectral radius of matrix $A$ is denoted $\sigma(A)$.
\section{The distributed Kalman filter}
\label{sec:statement}
\subsection{Statement of the problem}
Consider $M$ interconnected systems, each described by the following equations:
\begin{equation}
\label{eq:subsystems00}
\begin{array}{ll}
x_i(k+1)&=A_{ii}x_i(k)+\sum_{j\neq i}A_{ij}x_j(k)+w_i(k)\\
y_i(k)&=C_{i}x_i(k)+v_i(k)
\end{array}
\end{equation}
where $x_i(k),w_i(k)\in\mathbb{R}^{n_i}$ and $y_i(k),v_i(k)\in\mathbb{R}^{p_i}$.
We assume that $w_i(k)$ and $v_i(k)$ are zero-mean white noises, for all $i=1,\dots,M$, and that $\mathbb{E}\{w_i(k)w_j^T(k)\}=Q_{i}\delta_{ij}$, $\mathbb{E}\{v_i(k)v_j(k)\}=R_i\delta_{ij}$ (with $R_i>0$ for all $i=1,\dots,M$), and that $\mathbb{E}\{w_i(k)v_j^T(h)\}=0$ for all $i,j=1,\dots,M$, and $h,k\geq 0$. For $i \in \left\{1,\ldots, M\right\}$, we define with $\mathcal{N}_i$ the set of neighbors (also denoted predecessors in \cite{RiversoFarinaGFT_PnP13}) of subsystem $i$ defined as $\mathcal{N}_i=\left\{j \, |\, A_{ij}\neq 0\right\}$ while $\mathcal{S}_i$ is the set of successors of subsystem $i$ defined as $\mathcal{S}_i=\left\{j \, |\, i \in \mathcal{N}_j \right\}$. In our setup we assume that subsystem $i$ can exchange information with its neighbors. Note that $i$ is in general included in $\mathcal{S}_i$ and $\mathcal{N}_i$.\\
Collectively, if we define the variables $\mathbf{x}(k)=(x_1(k),\dots,x_M(k))$, $\mathbf{y}(k)=(y_1(k),\dots,y_M(k))$, $\mathbf{w}(k)=(w_1(k),\dots,w_M(k))$, and $\mathbf{v}(k)=(v_1(k),\dots,v_M(k))$, we can rewrite \eqref{eq:subsystems00} as
\begin{equation}
\label{eq:system00}
\begin{array}{ll}
\mathbf{x}(k+1)&=\mathbf{A}\mathbf{x}(k)+\mathbf{w}(k)\\
\mathbf{y}(k)&=\mathbf{C}\mathbf{x}(k)+\mathbf{v}(k)
\end{array}
\end{equation}
where $\mathbf{C}=$diag$(C_1,\dots,C_M)$, $\mathbf{Q}=$diag$(Q_1,\dots,Q_M)$, $\mathbf{R}=$diag$(R_1,\dots,R_M)$, and\\
$$\mathbf{A}=\begin{bmatrix}A_{11}&\dots&A_{1M}\\
\vdots&\ddots&\vdots\\
A_{M1}&\dots&A_{MM}\end{bmatrix}.$$\\
The optimal centralized Kalman predictor \cite{GoodwinRiccati-84} for system \eqref{eq:system00} is
\begin{equation}
\label{eq:predictor00}
\hat{\mathbf{x}}_c(k+1)=\mathbf{A}\hat{\mathbf{x}}_c(k)+\mathbf{L}_c(k)(\mathbf{y}(k)-\mathbf{C}\hat{\mathbf{x}}_c(k))
\end{equation}
where $\hat{\mathbf{x}}_c(k)$ denotes the one-step optimal predictor of $\mathbf{x}(k)$. According to the classical Kalman prediction theory, the optimal gain is
\begin{equation}
\mathbf{L}_c(k)=\mathcal{L}(\Pi_c(k),\mathbf{A},\mathbf{C},\mathbf{R})
\label{eq:riccati_centrgain}\end{equation}
where $\Pi_c(k)$ is the centralized Kalman prediction error covariance matrix and is computed iteratively using the Riccati equation
\begin{equation}\begin{array}{l}
\Pi_c(k+1)=\mathcal{R}(\Pi_c(k),\mathbf{A},\mathbf{C},\mathbf{Q},\mathbf{R})\\
\quad=(\mathbf{A}-\mathbf{L}(k)\mathbf{C})\Pi_c(k)(\mathbf{A}-\mathbf{L}(k)\mathbf{C})^T+\mathbf{Q}+\mathbf{L}(k)\mathbf{R}\mathbf{L}(k)^T
\end{array}\label{eq:riccati_centr00}\end{equation}
\subsection{Distributed prediction scheme}
As clear from \eqref{eq:predictor00}-\eqref{eq:riccati_centr00}, the optimal centralized Kalman predictor for system \eqref{eq:system00} is based on the iteration of the Riccati equation \eqref{eq:riccati_centr00}, which requires a global knowledge of the system and, in general, leads to a matrix gain which has not the sparsity properties of the dynamic system (i.e., of matrix $\mathbf{A}$).\\
In contrast, in this paper we seek for a distributed observer implementation, meaning that: (\emph{i}) at most data originated by neighbors are used by the local observers, to reduce the communication load of the scheme; (\emph{ii}) information about the model of the overall system is not required to be stored by each local observer, but at most information concerning the neighboring subsystems; (\emph{iii}) the computational load required by each local filter is scalable.\\
In line with this we propose an estimation scheme of the type
\begin{equation}
\label{eq:subpredictor00}
\hat{x}_i(k+1)=\sum_{j\in \mathcal{N}_i}\left\{A_{ij}\hat{x}_j(k)+ L_{ij}(k)(y_j(k)-C_j\hat{x}_j(k))\right\}
\end{equation}
where
\begin{equation}\label{eq:gain_distributed}
L_{ij}(k)=A_{ij}P_j(k)C_j^T(C_jP_j(k)C_j^T+R_j)^{-1}.
\end{equation}
The matrices $P_i(k)$, $i=1,\dots,M$ are updated according to the following distributed equation.
\begin{equation}
\label{eq:riccati_distr01}
\begin{array}{ll}
P_i(k+1)=&\sum_{j \in \mathcal{N}_i} \left(\tilde{A}_{ij}P_j(k)\tilde{A}_{ij}^T-\tilde{A}_{ij}P_j(k)\tilde{C}_j^T\right.\\
&\qquad\left.\cdot(\tilde{C}_jP_j(k)\tilde{C}_j^T+\tilde{R}_j)^{-1}\tilde{C}_jP_j(k)\tilde{A}_{ij}^T\right) + Q_i \end{array}
\end{equation}
where, for all $i,j=1,\dots,M$, we have defined $\tilde{A}_{ij}=\sqrt{\varsigma_j} A_{ij}$, $\tilde{C}_{i}=\sqrt{\varsigma_i} C_{i}$, and $\tilde{R}_{i}=\varsigma_i R_{i}$, where $\varsigma_i=|\mathcal{S}_i|$. Note that \eqref{eq:gain_distributed} and \eqref{eq:riccati_distr01} are equivalent to $L_{ij}(k)=\mathcal{L}(P_j(k),A_{ij},C_j,R_j)=\mathcal{L}(P_j(k),\tilde{A}_{ij},\tilde{C}_j,\tilde{R}_j)$ and $P_i(k+1)=\sum_{j\in\mathcal{N}_i}\mathcal{R}(P_j(k),\tilde{A}_{ij},\tilde{C}_j,0,\tilde{R}_j)+Q_i=
\sum_{j\in\mathcal{N}_i}\varsigma_j\mathcal{R}(P_j(k),{A}_{ij},{C}_j,0,{R}_j)+Q_i$, respectively.

Equation \eqref{eq:subpredictor00} is distributed, i.e., $L_{ij}\neq 0$ only if $A_{ij}\neq 0$. Therefore, the computation of $L_{ij}$ can be done distributedly and communication is required between local state estimators of dynamically interconnected subsystems only. Concerning the scalability of the algorithm observe also that, for $i \in \left\{1,\ldots, M\right\}$, subsystem $i$ permanently stores in memory only the matrices $Q_i$, $R_i$, $A_{ii}$, $C_i$ and, for $j \in \mathcal{N}_{i}$, $A _{ij}$, $C_j$ and $R_j$; on the other hand, the information which must be transmitted and temporarily stored at each time step consists of ${y}_j$, $\hat{x}_j$, $P_j$, $L_j^F$ for all $j \in \mathcal{N}_{i}$. The DKF algorithm is formally described in Algorithm~1.

\begin{algorithm}
\label{algo:DKF}
\caption{DKF algorithm}
\textbf{Memory requirements}\\
For $i \in \left\{1,\ldots,N\right\}$, subsystem $i$ stores in memory
\begin{algorithmic}
\State - Permanently $Q_i$, $R_i$, $A_{ii}$, $C_i$, $\left\{A _{ij}, C_j, R_j; j \in \mathcal{N}_{i} \right\}$;
\State - Temporarily, for each $k\geq 1$, $\left\{y_j(k), \hat{x}_j(k), P_j(k), L_{ij}(k); j \in \mathcal{N}_{i} \right\}$;
\end{algorithmic}
\hrulefill \\
\textbf{On-line implementation}\\
At each time step $k\geq 1$ subsystem $i$:
\begin{algorithmic}
\State \textbf{1)} Measures $y_i(k)$
\State \textbf{2)} Broadcasts to its successors the quantities $y_i(k)$, $\hat{x}_i(k)$, and $P_i(k)$;
\State \textbf{3)} Gathers from its neighbors the information $\left\{y_j(k),\hat{x}_j(k), P_j(k); j \in \mathcal{N}_{i} \right\}$;
\State \textbf{4)} Computes the gains $\left\{L_{ij}(k); j \in \mathcal{N}_{i} \right\}$ as in \eqref{eq:gain_distributed};
\State \textbf{5)} Computes the estimate $\hat{x}_i(k+1)$ and the matrix $P_i(k+1)$ as in \eqref{eq:subpredictor00} and \eqref{eq:riccati_distr01}, respectively.
\end{algorithmic}
\end{algorithm}
\subsection{Main properties}
Let us define $\hat{\mathbf{x}}(k)=(\hat{x}_1(k),\dots,\hat{x}_M(k))$, and the distributed filter estimation error $\mathbf{e}(k)=\mathbf{x}(k)-\hat{\mathbf{x}}(k)$. From \eqref{eq:system00} and \eqref{eq:subpredictor00} we obtain that
\begin{equation}\mathbf{e}(k+1)=(\mathbf{A}-\mathbf{L}(k)\mathbf{C})\mathbf{e}(k)-\mathbf{L}(k)\mathbf{v}(k)+\mathbf{w}(k)\label{eq:prediction_error_collective}\end{equation}
where ${\mathbf{L}}(k)$ is the matrix whose block entries are ${L}_{ij}(k)$. Let $\Pi_d(k)=$var$(\mathbf{e}(k))$. From~\eqref{eq:prediction_error_collective} the following is obtained.
\begin{equation}\Pi_d(k+1)=(\mathbf{A}-\mathbf{L}(k)\mathbf{C}) \Pi_d(k)(\mathbf{A}-\mathbf{L}(k)\mathbf{C})^T -\mathbf{L}(k)  \mathbf{R}  \mathbf{L}(k)^T+  \mathbf{Q}\label{eq:prediction_error_collective_var}\end{equation}
The following result can be derived.
\begin{lemma}
\label{lemma:stability}
Assume that the pair $(\mathbf{A,G})$ is stabilizable (where $\mathbf{GG}^T=\mathbf{Q}$) and that there exist symmetric matrices $\bar{P}_i\geq 0$, $i=1,\dots,M$ such that
\begin{equation}\label{eq:Pi_bar}
\bar{P}_i\geq \sum_{j\in\mathcal{N}_i}\mathcal{R}(\bar{P}_j,\tilde{A}_{ij},\tilde{C}_j,0,\tilde{R}_j)+Q_i,
\end{equation}
for all $i=1,\dots,M$. For all $i,j=1,\dots,M$, let
$\bar{L}_{ij}=\mathcal{L}(\bar{P}_j,A_{ij},C_j,R_j)$ and let $\bar{\mathbf{L}}$ be the matrix whose block entries are $\bar{L}_{ij}$. Then, the matrix $\mathbf{A}-\bar{\mathbf{L}}\mathbf{C}$ is Schur stable.\hfill{}$\square$
\end{lemma}
Thanks to Lemma~\ref{lemma:stability}, a \emph{simplified version} of the DKF Algorithm~1 can be devised: assuming that each subsystem $i$ stores in memory matrix $\bar{P}_i$, $i=1,\dots,M$, with property \eqref{eq:Pi_bar}, then it is sufficient to set $P_i(k)=\bar{P}_i$ and $L_{ij}(k)=\bar{L}_{ij}=\mathcal{L}(\bar{P}_j,A_{ij},C_j,R_j)$ for all $k$ to guarantee that the estimation error $\mathbf{e}(k)$ is a stationary process. Therefore, the error covariance of this modified scheme is asymptotically convergent to a bounded definite positive matrix, i.e., $\lim_{k \to \infty} \Pi_d(k) = \bar{\Pi}_d$ for some positive definite matrix $\bar{\Pi}_d$.

In case the Algorithm~1 is implemented, under the assumption that there exist steady-state solutions of \eqref{eq:riccati_distr01} $\bar{P}_i\geq 0$, $i=1,\dots,M$, the next result can be proved.
\begin{proposition}\label{prop:stability}
Consider the DKF Algorithm 1. Assume that $P_i(1)$, $i=1,\ldots,M$, are such that there exists $\bar{P}_i$ with the property that
\begin{equation}\label{eq:Pi_bar_Eq}
\lim_{k \to \infty} P_i(k) = \bar{P}_i.
\end{equation}
Let $\bar{\mathbf{P}}=$diag$(\bar{P}_1,\dots,\bar{P}_M)$. Then, there exists a positive definite matrix $\bar{\Pi}_d$ such that $\lim_{k\rightarrow \infty}\Pi_d(k)= \bar{\Pi}_d$ and $\bar{\Pi}_d\leq \bar{\mathbf{P}}$. \hfill{}$\square$
\end{proposition}
Note that, under the validity of \eqref{eq:Pi_bar_Eq}, then in steady state conditions also \eqref{eq:Pi_bar} is verified. Therefore the DKF Algorithm~1 provides a stationary equation error; also, Proposition~\ref{prop:stability} states that, for $i=1,\ldots,M$, matrix $\bar{P}_i$ plays the role of an upper bound of the covariance of the prediction error $x_i(k)-\hat{x}_i(k)$ in steady state.\\
Observe that Lemma \ref{lemma:stability} and Proposition \ref{prop:stability} require the existence of matrices $\bar{P}_i$, $i=1,\ldots,M$, such that either property \eqref{eq:Pi_bar} or property \eqref{eq:Pi_bar_Eq} are satisfied. However, differently from the centralized Kalman filter, these properties are not guaranteed by standard detectability assumptions on the system.\\
%
%
In this paper we provide conditions under which these properties can be verified. In particular, in Section \ref{sec:convergence} we discuss the conditions allowing the application of  centralized design procedures while, in Section \ref{sec:d-design}, we provide a distributed design procedure.\\
We conclude this section with a couple of remarks.
\begin{remark}
Consider the DKF algorithm and assume that \eqref{eq:Pi_bar_Eq} holds true.
Let $\bar{\Pi}_c$ be the steady-state covariance of the prediction error for the centralized Kalman filter. Obviously, if $\Pi_d(1)=\Pi_c(1)$ then  $\Pi_c(k)\leq \Pi_d(k)$ for all $k\geq1$.
\end{remark}
\begin{remark}\label{rem:subopt_costant}
Assume that \eqref{eq:Pi_bar} holds true and consider the \emph{simplified DKF algorithm} described after Lemma \ref{lemma:stability}. Then, also in this case, the asymptotic covariance of the prediction error $x_i(k)-\hat{x}_i(k)$ is
upper-bounded by the matrix $\bar{P}_i$.
\end{remark}

\section{Centralized design}
\label{sec:convergence}
In this section we address the problem of providing ($i$) conditions that can be used to guarantee a-priori the validity of properties \eqref{eq:Pi_bar} or \eqref{eq:Pi_bar_Eq} and ($ii$) practical methods for computing them. First, in Section \ref{subsec:convergence:LMI} we will analyze \eqref{eq:Pi_bar} through a linear matrix inequality approach; secondly, in Section \ref{subsec:convergence:smallgain} we will provide an aggregate design procedure, based on small gain arguments, to guarantee \eqref{eq:Pi_bar_Eq}
\subsection{Design using LMI's}
\label{subsec:convergence:LMI}
In this section we provide a practical method based on LMI's for computing, if possible, matrices $\bar{P}_i$ verifying \eqref{eq:Pi_bar}. Then, as already highlighted, if we set $P_i(k)=\bar{P}_i$ for all $k$ and for all $i=1,\dots,M$, then it is guaranteed that this simplified version of the DKF algorithm has suitable convergence properties in view of Lemma \ref{lemma:stability}. Also, its suboptimality features are discussed in Remark \ref{rem:subopt_costant}.\\
Using LMI's we aim to compute (see Lemma \ref{lemma:stability}) $\bar{P}_i$, $i=1,\dots, M$ verifying
\begin{align}\label{eq:riccati_distr04_alg}
\bar{P}_i&\geq \sum_{j \in \mathcal{N}_i} \left(\tilde{A}_{ij}\bar{P}_j\tilde{A}_{ij}^T-\tilde{A}_{ij}\bar{P}_j\tilde{C}_j^T(\tilde{C}_j\bar{P}_j\tilde{C}_j^T+\tilde{R}_j)^{-1}\tilde{C}_j\bar{P}_j\tilde{A}_{ij}^T\right) + Q_i
\end{align}
Provided that $\bar{P}_j$ is non singular for each $j\in\mathcal{N}_i$, the algebraic inequality \eqref{eq:riccati_distr04_alg} is equivalent to
\begin{equation}\label{eq:riccati_distr05_alg}
\begin{array}{ll}
\bar{P}_i&\geq \sum_{j \in \mathcal{N}_i} \tilde{A}_{ij}\bar{P}_j(\bar{P}_j+\bar{P}_j\tilde{C}_j^T\tilde{R}_{j}^{-1}\tilde{C}_j\bar{P}_j)^{-1}\bar{P}_j\tilde{A}_{ij}^T + Q_i\\
&\geq\sum_{j \in \mathcal{N}_i} \tilde{A}_{ij}(\bar{P}_j^{-1}+\tilde{C}_j^T\tilde{R}_j^{-1}\tilde{C}_j)^{-1}\tilde{A}_{ij}^T + Q_i
\end{array}\end{equation}
thanks to the application of the matrix inversion lemma. Inequality \eqref{eq:riccati_distr05_alg} can be cast as the following LMI
\begin{equation}
\label{eq:LMI01}
\begin{bmatrix}
\bar{P}_i&\tilde{A}_{i1}\Delta_{1}&\dots&\tilde{A}_{iM}\Delta_M&G_i\\
*&\Delta_{1}&\dots&0&0\\
\vdots&\vdots&\ddots&\vdots&\vdots\\
*&0&\dots&\Delta_M&0\\
*&0&\dots&0&I
\end{bmatrix}\geq 0
\end{equation}
where $G_i$ is defined in such a way that $G_iG_i^T=Q_i$ and, for all $j=1,\dots,M$, $\Delta_j\geq (\bar{P}_j^{-1}+\tilde{C}_j^T\tilde{R}_j^{-1}\tilde{C}_j)^{-1}$. If we define $\Omega_j=\bar{P}_j^{-1}$, the latter inequality can be written as
\begin{equation}
\label{eq:LMI02}
\begin{bmatrix}
\Delta_j&I\\
I&\Omega_j+\tilde{C}_j^T\tilde{R}_j^{-1}\tilde{C}_j
\end{bmatrix}\geq 0.
\end{equation}
Finally, the equality $\Omega_j=\bar{P}_j^{-1}$ can be managed using the recursive approach proposed in \cite{ElGhaoui_ConeComplementarity}.
Indeed, we solve the following LMI
\begin{equation}
        \begin{bmatrix}\Omega_j& I\\
        I& \bar{P}_j
         \end{bmatrix}\geq 0 \label{eq:Cone_LMI}
\end{equation}
and, at the same time, we minimize the additional cost function
$\mathrm{tr}\{\Omega_j \bar{P}_j\}$.
The problem can be managed using the recursive cone complementarity linearization algorithm discussed in \cite{ElGhaoui_ConeComplementarity}.
%
%
\subsection{Design using small gain arguments}
\label{subsec:convergence:smallgain}
In this section we investigate conditions ensuring the validity of \eqref{eq:Pi_bar_Eq}. In particular,
the following result addresses the offline design issue providing an aggregate and lightweight analytical condition, which relies on small-gain arguments. First, the following assumption is required.
\begin{assumption}
\label{ass:local-in} For subsystem $i$, ${A}_{ii}$ is invertible.\hfill$\square$
\end{assumption}
We will also need one of the following assumptions for properly initializing $P_i(1)$ for the implementation of Algorithm~1.
\begin{assumption}
\label{ass:local-det1}For subsystem $i$
\begin{itemize}
\item[(i)] $({A}_{ii},{C}_i)$ is detectable;
\item[(ii)] $({A}_{ii},G_i)$ is stabilizable, where $G_i$ verifies $G_iG_i^T=Q_i$;
\end{itemize}
\end{assumption}
\begin{assumption}
\label{ass:local-det2}For subsystem $i$
\begin{itemize}
\item[(i)] $(\tilde{A}_{ii},\tilde{C}_i)$ is detectable;
\item[(ii)] $(\tilde{A}_{ii},G_i)$ is stabilizable;
\end{itemize}
\end{assumption}
Note that, while Assumption~\ref{ass:local-det1} is required to define, for a given subsystem $i$, $\bar{P}_i^N$ as the unique semi-positive definite solution to the local Riccati algebraic equation $\bar{P}_i^N=\mathcal{R}(\bar{P}_i^N,{A}_{ii},{C}_i,{Q}_i,{R}_i)$, Assumption~\ref{ass:local-det2} allows to define $\tilde{P}_i^N$ as the unique semi-positive definite solution to $\tilde{P}_i^N=\mathcal{R}(\tilde{P}_i^N,\tilde{A}_{ii},\tilde{C}_i,{Q}_i,\tilde{R}_i)$.

Let us now define full rank $n_i$ arbitrary transformation matrices $H_i$, $i=1,\dots,M$, i.e., $H_i\in\mathbb{R}^{n_i\times n_i}$ (introduced for reducing, if possible, the conservativity of the results stated next) and gains $\bar{L}_i$, selected in such a way that $\bar{F}_i=\tilde{A}_{ii}-\bar{L}_i \tilde{C}_i$ is Schur stable.
Define also $\hat{F}_i=H_i\bar{F}_iH_i^{-1}$ and $\hat{A}_{ij}=H_i {A}_{ij}H_j^{-1}$, for all $j=1,\dots,M$.
Finally we define
$${\Gamma}=\left\{
\begin{array}{ll}{\gamma}_{ij}=0&\text{if }j=i\\
{\gamma}_{ij}=\frac{\mu_i^2}{1-\lambda_i^2}\|\hat{A}_{ij}\hat{A}_{jj}^{-1}\|^2 & \text{if }j\neq i
\end{array}
\right.$$
Scalars $\mu_i\geq 1$, $\lambda_i\in[0,1)$ are defined in such a way that $\|\hat{F}_i^h\|\leq \mu_i \lambda_i^h$. We introduce a further assumption.
\begin{assumption}
For some values of $\bar{L}_i$, $H_i$, (i) $\sigma(\hat{F}_i)<1$ for all $i=1,\dots,M$, and (ii) $\sigma(\Gamma)<1$.
\label{ass:small-gain-global}
\end{assumption}
Note that, a necessary condition for the existence of matrix $\bar{L}_i$ guaranteeing that $\sigma(\hat{F}_i)<1$ is that $(\tilde{A}_{ii},\tilde{C}_i)$ is detectable, i.e., Assumption~\ref{ass:local-det2}; therefore, the latter is implicitly required by Assumption~\ref{ass:small-gain-global}.
\begin{theorem}
\label{prop:small-gain}
 If Assumption \ref{ass:local-in} holds for all $i=1,\dots,M$ and under Assumption \ref{ass:small-gain-global}, there exist $\bar{P}_i\geq 0$ for all $i=1,\dots,M$ such that, $P_i(k)\rightarrow \bar{P}_i$ as $k\rightarrow\infty$ if one of the following initializations is used:
\begin{enumerate}[a.]
\item $P_i(1)=0$ for all $i=1,\dots,M$.
\item $P_i(1)=\bar{P}_i^N$ if Assumption~\ref{ass:local-det1} holds for all $i=1,\dots,M$.
\item $P_i(1)=\tilde{P}_i^N$ if Assumption~\ref{ass:local-det2} holds for all $i=1,\dots,M$.\hfill{}$\square$
\end{enumerate}
\end{theorem}
Regarding Assumption~\ref{ass:small-gain-global}, provided that Assumption~\ref{ass:local-det2} (i) is verified, it is always possible of find $\bar{L}_i$ such that $\sigma(\hat{F}_i)<1$ for all $i=1,\dots,M$. Note that, in case the system has a cascade topology (i.e., if it admits a lower - or upper - block triangular form,~\cite{Siljak91}), $\Gamma$ is block triangular, and therefore Assumption ~\ref{ass:small-gain-global} can be easily verified.\\
On the other hand, for more general system structures, we need to retrieve a suitable ``decentralized" change of coordinates and, at the same time, a suitable ``auxiliary" decentralized linear observer, for which $\sigma(\hat{F}_i)<1$ for all $i=1,\dots,M$ and the corresponding matrix $\Gamma$ is stable. This amounts to a design problem, which can be cast, for example, as the following optimization.
\begin{subequations}
\label{eq:centralized_design}
\begin{align}
\min_{\{H_{i},L_{i}\}_{i=1}^M} \sigma(\Gamma)
\end{align}
subject to the definition of $\Gamma$ and to
\begin{align}
\sigma(\Gamma)&<1\\
\sigma(\hat{F}_i)&<1,\,\,i=1,\dots,M
\end{align}
\end{subequations}
To reduce the computational load of \eqref{eq:centralized_design} the values of $H_i$ can be constrained. For example, one can try to minimize the terms $\mu_i$ by constraining $H_i$ to take values corresponding to which $\hat{F}_i=H_i\bar{F}_iH_i^{-1}$ is diagonal (provided that $\bar{F}_i$ is diagonalizable and has real eigenvalues), or one can set $H_i=I_{n_i}$. The optimization \eqref{eq:centralized_design} is nonlinear, and therefore a suitable initialization is fundamental, for example selecting $\bar{F}_i$ as the Kalman predictor gains.\\
To reduce the computational complexity and to allow for flexible and reliable operation, in next Section we provide a distributed and scalable design procedure to be applied at each subsystem level.
\section{Distributed design and plug and play features}
\label{sec:d-design}
In many practical applications, it is of interest to perform the design of the DKF in a distributed fashion, i.e., to have a set of conditions to be verified locally by each subsystem, possibly using pieces of information provided by the neighboring subsystems.\\
Focusing on the main assumptions of Theorem \ref{prop:small-gain}, while Assumptions~\ref{ass:local-in} and~\ref{ass:small-gain-global} (i) are local conditions, to be verified at a single subsystem level, Assumption \ref{ass:small-gain-global} (ii) is centralized (although aggregate), since it involves information concerning the overall system.
We now introduce the following assumption, providing a conservative, yet distributed and very simple, condition, that must be verified at a single subsystem level by each subsystem, that implies the Schur stability of $\Gamma$, as proved in Proposition~\ref{propo:Gamma-distributed} stated below.
\begin{assumption}
For all $i=1,\dots, M$ and for some values of $\bar{L}_i$, $H_i$, it holds that
\begin{equation}\rho_i=\sum_{j=1}^M\gamma_{ij}<1\end{equation}
\label{ass:small-gain-distributed}
\end{assumption}
\begin{proposition}
\label{propo:Gamma-distributed}
If Assumption \ref{ass:small-gain-distributed} holds, then Assumption \ref{ass:small-gain-global} (ii) is verified.\hfill$\square$
\end{proposition}
As it will be shown in the remainder of the section, this result allows for PnP operation.
The PnP scenario consists of the case when one or more subsystems (each described by \eqref{eq:subsystems00}) or devices (and specifically a transducer) is added to or removed from the interconnected system.\\
Before to proceed, the following standing assumption sets the scenario where PnP operations take place, assuming that the PnP event occurs at time instant $k=T_{PnP}$.
\noindent
\begin{assumption}\hfill\\
- For $k<T_{PnP}$, Assumptions~\ref{ass:small-gain-global} (i),~\ref{ass:small-gain-distributed}, and~\ref{ass:local-in} (for all $i=1,\dots,M$) hold.\\
- At $k=T_{PnP}$ the updates \eqref{eq:riccati_distr01} are in steady state, i.e., $P_i(T_{PnP})=\bar{P}_i$ for all $i=1,\dots,M$.\hfill$\square$
\label{ass:PnP}
\end{assumption}
It is important to remark that, when PnP operations involving subsystems take place, the number of successors, for some subsystems, may change. Denote with ${\mathcal{S}}_i^+$ the set of successors of subsystem $i$ after the PnP event and $\varsigma_i^+=|{\mathcal{S}}_i^+|$. In general it holds that $\varsigma_i^+\neq \varsigma_i$. From this, it also follows that the matrices $\tilde{A}_{ij}$, $\tilde{C}_{i}$, and $\tilde{R}_i$ must be redefined, i.e., $\tilde{A}^+_{ij}=\sqrt{{\varsigma_j^+}}{A}_{ij}=\sqrt{\frac{\varsigma_j^+}{\varsigma_j}}\tilde{A}_{ij}$, $\tilde{C}^+_{i}=\sqrt{{\varsigma_i^+}}{C}_{i}=\sqrt{\frac{\varsigma_i^+}{\varsigma_i}}\tilde{C}_{i}$, and $\tilde{R}^+_i=\varsigma_i^+{R}_i=\frac{\varsigma_i^+}{\varsigma_i}\tilde{R}_i$.
Importantly, in case $\varsigma_i^+>\varsigma_i$, this may prevent the detectability of the pair $(\tilde{A}_{ii}^+,\tilde{C}^+_{i})$ to hold, which may jeopardize the verifiability of Assumption~\ref{ass:small-gain-global} (i). We also assume that $H_i$ and $\bar{L}_i$ are not redefined, for $i=1,\dots,M$ after the PnP event. From this it follows that, for all $i,j=1,\dots,M$, $\|\hat{A}_{ij}^+(\hat{A}_{jj}^+)^{-1}\|=\|\hat{A}_{ij}\hat{A}_{jj}^{-1}\|$.
\subsection{Plug-in of a subsystem}
Assume that, at step $T_{PnP}$, the subsystem $(M+1)$ is introduced. For each $i=1,\dots,M$
\begin{itemize}
\item if $i\not\in\mathcal{S}_{M+1}\cup\mathcal{N}_{M+1}$, then $\gamma_{i(M+1)}=0$. Also, since $\varsigma_i^+=\varsigma_i$, $\hat{F}_i^+=\hat{F}_i$, then $\mu_i^+=\mu_i$ and $\lambda_i^+=\lambda_i$. In view of this, $\rho_i^+=\sum_{j=1}^{M+1}\gamma^+_{ij}=\sum_{j=1}^{M}\gamma_{ij}=\rho_i<1$;
\item if $i\in\mathcal{S}_{M+1}$ but $i\not\in\mathcal{N}_{M+1}$, then $\varsigma_i^+=\varsigma_i$: Therefore $\hat{F}_i^+=\hat{F}_i$, $\mu_i^+=\mu_i$, and $\lambda_i^+=\lambda_i$. However, since $M+1\in\mathcal{N}_i$, $\gamma_{i(M+1)}>0$. Therefore
$\rho_i^+=\sum_{j=1}^{M+1}\gamma^+_{ij}=\rho_i+\gamma^+_{i(M+1)}>\rho_i$;

\item if $i\in\mathcal{N}_{M+1}$ but $i\not\in\mathcal{S}_{M+1}$, then $\gamma_{i(M+1)}=0$. However, $\varsigma_i^+=\varsigma_i+1$ and therefore $\bar{F}_i^+=\sqrt{\varsigma_i+1}(A_{ii}-\bar{L}_iC_i)=\sqrt{(\varsigma_i+1)/\varsigma_i}\bar{F}_i$ may not be Schur stable. If $\bar{F}_i^+$ is stable, $\mu_i^+=\mu_i$ but, at the same time, $\lambda_i^+=\sqrt{(\varsigma_i+1)/\varsigma_i}\lambda_i>\lambda_i$. In view of this $\gamma_{ij}^+=(1-\lambda_i^2)/(1-\lambda_i^{+2})\gamma_{ij}$ for all $j=1,\dots,M$. Therefore
$\rho_i^+=(1-\lambda_i^2)/(1-\lambda_i^{+2})\rho_i>\rho_i$;
\item $i\in\mathcal{S}_{M+1}\cap\mathcal{N}_{M+1}$, the Schur stability of $\bar{F}_i^+$ is not guaranteed. If $\bar{F}_i^+$ is Schur stable, we can compute $\rho_i^+=\sum_{j=1}^{M}\gamma^+_{ij}+\gamma^+_{i(M+1)}=
    (1-\lambda_i^2)/(1-\lambda_i^{+2})\rho_i+\gamma^+_{i(M+1)}>\rho_i$, in view of the fact that both $\lambda_i^+>\lambda_i$ and $\gamma^+_{i(M+1)}>0$.
\end{itemize}
The design of $\bar{L}_{M+1}$, $H_{M+1}$ can be addressed through the following optimization problem.
\begin{subequations}
\label{eq:scalable_design}
\begin{align}
\min_{H_{M+1},L_{M+1}} \rho_{M+1}^+ + \sum_{j\in \mathcal{S}_{M+1}}\gamma_{j(M+1)}^+
\end{align}
subject to:
\begin{align}
\sigma(\hat{F}_{M+1}^+) < 1&
,\,\rho_{M+1}^+ < 1\\
\gamma_{j(M+1)}^+<1-\frac{1-\lambda_j^2}{1-\lambda_j^{+2}}\rho_j&\text{ for all } j\in \mathcal{S}_{M+1}
\end{align}
\end{subequations}
When a plug-in request is received from subsystem $M+1$, the following design procedure must be adopted: (i) if \eqref{eq:scalable_design} admits a solution and if, for all $i\in\mathcal{N}_{M+1}$, $\rho_{i}^+<1$ and $\sigma(\hat{F}_i)<1$, then allow the plug-in, otherwise deny it; (ii) properly initialize $P_{M+1}(T_{PnP})$.\\
The following corollary of Theorem \ref{prop:small-gain} addresses the step (ii) and guarantees convergence of the system matrices $P_i(k)$, $k=1,\dots,M+1$ to steady state solutions.
\begin{corollary}
\label{cor:small-gain-Plugin}
If Assumption \ref{ass:local-in} holds also for $i=M+1$ and if, after the plug-in event, Assumptions \ref{ass:small-gain-global} (i) and \ref{ass:small-gain-distributed} are verified,
then there exist $\bar{P}_i^+\geq 0$ for all $i=1,\dots,M$ such that, $P_i(k)\rightarrow \bar{P}_i^+$ as $k\rightarrow\infty$ if the following initialization is used: $P_i(T_{PnP})=\bar{P}_i$ for all $i=1,\dots,M$ and (a) $P_{M+1}(T_{PnP})=0$, or (b) if Assumption~\ref{ass:local-det1} holds for $i=M+1$, $P_{M+1}(T_{PnP})=\bar{P}_{M+1}^N$, or (c) if Assumption~\ref{ass:local-det2} holds for $i=M+1$, $P_{M+1}(T_{PnP})=\tilde{P}_{M+1}^N$.\hfill$\square$
\end{corollary}
Note that the initializations (b) and (c) limit possible undesirable transients on the state estimates.
Note also that, at the $(M+1)$-th subsystem level, to solve \eqref{eq:scalable_design}, the required data consist in (i) the local system matrices $(A_{(M+1)(M+1)}, C_{M+1})$, (ii) the number $\varsigma_{M+1}$ of successors of subsystem $M+1$, (iii) $A_{(M+1)j}$, $A_{jj}$, $H_j$ for all $j\in\mathcal{N}_{M+1}$, (iv) $(1-\lambda_j^2)/(1-\lambda_j^{+2})\rho_j$, $A_{j(M+1)}$, $H_j$ for all $j\in \mathcal{S}_{M+1}$. It is therefore clear that this local design problem requires the transmission of a limited amount of information, i.e., through a neighbor-to-neighbor communication graph.\\
Also, remark that the optimization problem \eqref{eq:scalable_design} is a nonlinear one; to simplify it, an efficient strategy amounts, for example, to define $H_{M+1}$ as the matrix such that $\hat{F}_{M+1}$ is diagonal (i.e., in case $\hat{F}_{M+1}$ is diagonalizable and has real eigenvalues), making $H_{M+1}$ depend upon $\bar{L}_{M+1}$, or simply setting $H_{M+1}=I_{n_i}$. In this way we can reduce the number of free variables of the problem.
\subsection{Unplug of a subsystem}
\label{sec:unplug}
Assume that, without loss of generality, at step $T_{PnP}$, subsystem $M$ is unplugged. Note that,
\begin{itemize}
\item if $i\not\in\mathcal{S}_{M}\cup\mathcal{N}_{M}$, then $\rho_i^+=\sum_{j=1}^{M-1}\gamma_{ij}=\sum_{j=1}^{M}\gamma_{ij}=\rho_i<1$;
\item if $i\in\mathcal{S}_{M}$ but $i\not\in\mathcal{N}_{M}$, then $\rho_i^+=\sum_{j=1}^{M-1}\gamma_{ij}=\rho_i-\gamma_{iM}<1$;
\item if $i\in\mathcal{N}_{M}$ but $i\not\in\mathcal{S}_{M}$, then
$\varsigma_i^+=\varsigma_i-1$ and therefore $\bar{F}_i^+=\sqrt{(\varsigma_i-1)/\varsigma_i}\bar{F}_i$. From this it follows that $\mu_i^+=\mu_i$ but, at the same time, $\lambda_i^+=\sqrt{(\varsigma_i-1)/\varsigma_i}\lambda_i<\lambda_i$. Also $\gamma_{ij}^+=(1-\lambda_i^2)/(1-\lambda_i^{+2})\gamma_{ij}$ for all $j=1,\dots,M$. Therefore
$\rho_i^+=(1-\lambda_i^2)/(1-\lambda_i^{+2})\rho_i<\rho_i$;
\item $i\in\mathcal{S}_{M}\cap\mathcal{N}_{M}$, it follows that $\rho_i^+=\sum_{j=1}^{M-1}\gamma^+_{ij}=
    (1-\lambda_i^2)/(1-\lambda_i^{+2})(\rho_i-\gamma_{iM})<\rho_i$, in view of the fact that both $\lambda_i^+<\lambda_i$ and $\gamma^+_{iM}=0$.
\end{itemize}
In view of this, since Assumptions \ref{ass:small-gain-distributed} and~\ref{ass:small-gain-global} (i) hold before the unplug event, then they are guaranteed for the system deprived of the $M$-th subsystem. Therefore, any unplug request can be accepted, without hampering the convergence properties of the estimator.\\
The following corollary of Theorem \ref{prop:small-gain} guarantees convergence of the system matrices $P_i(k)$, $k=1,\dots,M-1$ to new steady state solutions.
\begin{corollary}
\label{cor:small-gain-unPlug}
After the un-plug event, there exist $\bar{P}_i^+\geq 0$ for all $i=1,\dots,M$ such that, $P_i(k)\rightarrow \bar{P}_i^+$ as $k\rightarrow\infty$ if the following initialization is used: $P_i(T_{PnP})=\bar{P}_i$ for all $i=1,\dots,M-1$.
\end{corollary}
\subsection{Plug and play of transducers}
In many practical applications, the sensors embedded in a subsystem can be added, removed, or replaced. We consider that changes occur to the $M$-th subsystem for simplicity, but without loss of generality. Practically, this case consists in a change in the matrix $C_M$ (and, consequently, $\tilde{C}_{M}$), while the topology of the system and its dynamics remain unchanged. Therefore, for all $i\neq M$, $\rho_i^+=\rho_i$, since $\rho_i$, $i\neq M$, do not depend on $C_M$, but only on matrices $A_{ij}$ and on the number of successors, which remain unchanged.\\
On the other hand, focusing on subsystem $M$
\begin{itemize}
\item if a transducer is plugged in, this consists of adding a row (here denoted $c_{add}$) to matrix $C_M$, i.e.,
$$C_M^+=\begin{bmatrix}C_M\\c_{add}
\end{bmatrix}$$
This means that the detectability properties of the pairs $(A_{ii},C_i)$ and $(\tilde{A}_{ii},\tilde{C}_i)$ are not jeopardized by the plug-in event. Also, if $H_M$ remains unchanged and if we take $\bar{L}_M^+=\begin{bmatrix}\bar{L}_M&0\end{bmatrix}$, then $\rho_M^+=\rho_M$. This means that the addition of a new transducer does not compromize the convergence properties of the DKF scheme.
\item if a sensor is replaced or unplugged, this consists of a substantial variation of the matrix $C_M$. This means that, before to allow the PnP operation, one must verify the existence of a gain $\bar{L}_M$ such that the following are verified: (I) Schur stability of $\bar{F}_i$; (II) $\rho_M^+<1$. Concerning the latter, note that $H_M$ should remain unchanged, in order not to affect the values of $\rho_i$, $i\neq M$.
\end{itemize}
In the case considered in this section, however, it is not clear how the PnP operation impacts on the values of the matrices $P_i(k)$, $i=1,\dots, M$. In order to guarantee the convergence of the matrices $P_i(k)$ to a steady state, the following practical procedure can be adopted, suggested by Corollaries \ref{cor:small-gain-Plugin} and \ref{cor:small-gain-unPlug}: (a) for $k\geq T_{PnP}$, make $P_i(k)$, $i=1,\dots,M-1$ evolve as if subsystem $M$ were unplugged; (b) after convergence is achieved (say, at instant $T_{conv}$) make $P_i(k)$, $i=1,\dots,M$, $k>T_{conv}$ evolve as if subsystem $M$ were plugged-in at time $T_{conv}$, i.e., by setting $P_i(T_{conv}+1)=P_i(T_{conv})$, $i=1,\dots,M$ and $P_M(T_{conv}+1)=\tilde{P}_{M}^N$.
\section{Simulation results}
\label{sec:Exs}
In this section we provide some simulation results illustrating the application of DKF to two different examples, an academic one and the Hycon2 benchmark described in \cite{SR-GFT:12}, respectively.
\subsection{Academic example}
In this section we consider a system composed of interconnected subsystems. We set\\
$$
A_{ii}=\begin{bmatrix}0.9&0.1\\0.1&-0.9\end{bmatrix}\qquad \text{and} \qquad
C_i=\begin{bmatrix}1&1
\end{bmatrix}$$
and, for all $j\in\mathcal{N}_i$, $A_{ij}=$diag$(\alpha,-\alpha)$, where $\alpha>0$. Also $\mathbb{E}\left[w_i w_i^T\right]=Q_i=I_{2}$ and $\mathbb{E}\left[v_i v_i^T\right]=R_i=1$ where $I_{2}$ is the $2$-dimensional identity matrix.
\subsubsection{Dependence on coupling and centralized design}\hfill\\
First consider $M=2$, with $\mathcal{N}_1=\mathcal{N}_2=\{1,2\}$.
In Figure \ref{fig:rho} we show the relationship between the coupling strength $\alpha\in[0,6.5]$ and (\emph{i}) the spectral radius $\sigma(\Gamma)$ of matrix $\Gamma$; (\emph{ii}) the spectral radius of $\mathbf{A}-\bar{\mathbf{L}}\mathbf{C}$ obtained through the LMI-based design procedure sketched in Section \ref{subsec:convergence:LMI}. The latter procedure has given numerically reliable results for $\alpha\leq 6.5$.\\
To realize the upper plot, two different choices of $H_{i}$ are adopted: (I) such that $\hat{F}_{M+1}$ is diagonal; (II) $H_{i}=I_2$ for all $i=1,2$. As it is apparent, the spectral radius of $\Gamma$ does not significantly vary in the latter cases.\\
In both cases it is apparent that the small-gain procedure sketched in Section \ref{subsec:convergence:smallgain} is applicable when the coupling strength is sufficiently small. Also, from the lower panel it is apparent that $\sigma(\Gamma)<1$ is just sufficient to guarantee that $\sigma(\mathbf{A}-\bar{\mathbf{L}}\mathbf{C})<1$.
\begin{figure}
\centering
\includegraphics[width=0.5\textwidth]{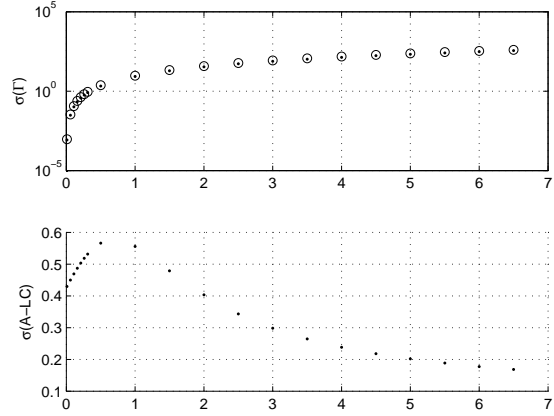}
\caption{Upper panel: values of $\sigma(\Gamma)$ as a function of $\alpha$ (setting $H_{i}$ such that $\hat{F}_{i}$ is diagonal - dots; setting $H_{i}=I_2$ - circles); lower panel: values of $\sigma(\mathbf{A}-\bar{\mathbf{L}}\mathbf{C})$ as a function of $\alpha$.}
\label{fig:rho}
\end{figure}

\subsubsection{Plug and play scenario}\hfill\\
Now, assume that $M=3$ and $\alpha=0.1$. At time $t=0$, $\mathcal{N}_1=\mathcal{N}_2=\{1,2\}$ and $\mathcal{N}_3=\{3\}$, i.e., subsystem 3 is not connected with the network. In this case
$$\Gamma=\begin{bmatrix}         0 &   0.1334 &        0\\
    0.1334   &      0   &      0\\
         0   &      0   &      0
\end{bmatrix}$$
Therefore, we have that $\rho_i<1$ for $i=1,2,3$.\\
At $t=100$ subsystem 3 plugs in and connects with subsystem 2. More specifically
$\mathcal{N}_1=\{1,2\}$, $\mathcal{N}_2=\{1,2,3\}$, and $\mathcal{N}_3=\{2,3\}$. In this case
$$\Gamma=\begin{bmatrix}         0 &   0.1334 &        0\\
    0.1535   &      0   &      0.1535\\
         0   &      0.0976   &      0
\end{bmatrix}$$
The plug-in request is accepted, since $\rho_i<1$ for $i=1,2,3$.\\
Finally, at $t=200$, subsystem 1 unplugs, meaning that $\mathcal{N}_1=\{1\}$ and $\mathcal{N}_2=\mathcal{N}_3=\{2,3\}$. As discussed in Section~\ref{sec:unplug}, the unplug request is automatically accepted, as it is witnessed by the values taken by the entries on $\Gamma$ in this case: $$\Gamma=\begin{bmatrix}         0 &   0 &        0\\
    0   &      0   &      0.1334\\
         0   &      0.0976   &      0
\end{bmatrix}$$
In Figure~\ref{fig:varspnp} the state trajectories are depicted, showing the different collective dynamical behaviours taken in correspondence with the different graph configurations.

\begin{figure}
\centering
\includegraphics[width=0.5\textwidth]{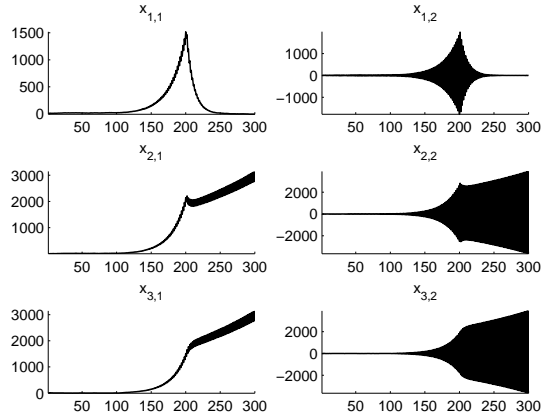}
\caption{State trajectories. $x_{i,k}$ denotes the $k$-th entry of $x_i$.}
\label{fig:varspnp}
\end{figure}

In Figure~\ref{fig:errspnp} the trajectories of the root mean estimation errors rmse$_i=\sqrt{1/n_s\sum_{t=1}^{n_s}\|x_i(t)-\hat{x}_i(t)\|^2}$ for all subsystems' states are depicted, showing that, in view of the proper matrix initializations, when plug and play operations occur the estimation error does not suffer from undesirable transients.
\begin{figure}
\centering
\includegraphics[width=0.5\textwidth]{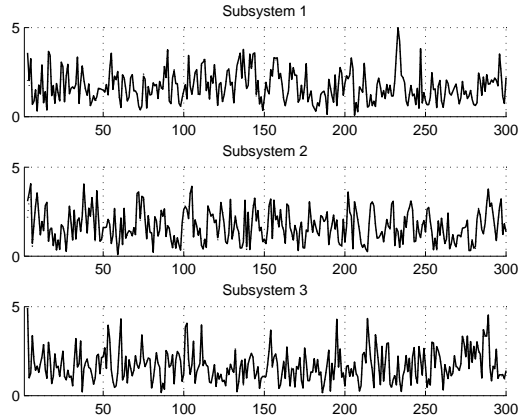}
\caption{Root mean estimation error for each subsystem, obtained with DKF - solid line - and centralized KF - dotted line (the lines are practically overlapping).}
\label{fig:errspnp}
\end{figure}

\subsection{Power network benchmark}
In this section we consider a power network system including a number of power generation areas coupled through tie-lines. This system has been adopted also in \cite{Negenborn-Kalman13} where the authors proposed a partition-based distributed estimation scheme tailored to power networks applications and exhibiting promising numerical results (although without any theoretical guarantees).\\
Our contributions are two-fold: firstly, in Section \ref{ex:comparison} we compare DKF with the centralized Kalman filter and the distributed strategy proposed in \cite{Negenborn-Kalman13}; secondly, in Section \ref{ex:PlugPlay} we test the PnP features of DKF in case a new subsystem is plugged in the network during its operation.\\
The dynamics of each power generation area, equipped with primary control and linearized around the equilibrium value for all variables, is described by the following continuous time LTI model \cite{SR-GFT:12}
\begin{equation}\label{eq:SysTurbine}
\dot{x}_i(t)=A^c_{ii}x_i(t)+B^c_i u_i +L^c_i \Delta P_{L_i} + \sum_{j \in \mathcal{N}_i} A^c_{ij} x_j
\end{equation}
where $x_i=(\Delta \theta_i, \Delta \omega_i, \Delta P_{m_i}, \Delta P_{v_i})$ is the state, $u_i=\Delta P_{ref_i}$ is the control input of each area, and $\Delta P_{L_i}$ is the local power load. Note that the letter $\Delta$ is used to denote the deviation from steady-state. The matrices of system \eqref{eq:SysTurbine} are
$$
A^c_{ii}=\left[
\begin{array}{cccc}
0 & 1 & 0 & 0 \\
-\frac{\sum_{j \in \mathcal{N}_i} P_{ij}}{2 H_i} & -\frac{D_i}{2 H_i} & \frac{1}{2 H_i} & 0 \\
0 & 0 & -\frac{1}{T_{t_i}} &  \frac{1}{T_{t_i}} \\
0 & -\frac{1}{R_i T_{g_i}} & 0 & - \frac{1}{T_{g_i}}
\end{array}
\right], \qquad
B_i^c = \left[
\begin{array}{c}
0\\
0\\
0\\
\frac{1}{T_{g_i}}
\end{array}
\right]
$$
$$
A^c_{ij}=\left[
\begin{array}{cccc}
0 & \,0\, & \,0 \, & \,0\, \\
\frac{P_{ij}}{2 H_i} & 0 & 0 & 0 \\
0 & 0 & 0 &  0 \\
0 & 0 & 0 & 0
\end{array}
\right], \qquad
L_i^c = \left[
\begin{array}{c}
0\\
-\frac{1}{2H_i}\\
0\\
0
\end{array}
\right]
$$
where the parameters and their numerical values are defined in \cite{SR-GFT:12}.
Since both $\Delta P_{ref_i}$ and $\Delta P_{L_i}$ are assumed to be constant and known, for the sake of simplicity, we neglect them from our analysis.\\
We discretize the process \eqref{eq:SysTurbine} with a sampling interval $T$ according to the technique proposed in \cite{FarinaAut2013}, leading to the discrete-time model \eqref{eq:subsystems00} where the matrices $A_{ii}$, $A_{ij}$ can be easily constructed from \eqref{eq:SysTurbine}. The matrix $C_i$ is
$$
C_{i}=\left[
\begin{array}{cccc}
1 & \,0\, & \,0 \, & \,0\, \\
0 & 1 & 0 & 0
\end{array}
\right]
$$
For $i \in \left\{1,\ldots,M\right\}$, $\mathbb{E}\left[w_i w_i^T\right]=Q_i=3\,I_{4}$ and $\mathbb{E}\left[v_i v_i^T\right]=R_i=I_{2}$ where $I_{k}$ is the $k$-dimensional identity matrix.

\subsubsection{Comparison test}\label{ex:comparison}\hfill\\
In this section we consider the scenario $1$ in \cite{SR-GFT:12}, where $M=4$ and where the adjacency matrix $Ad$, defining the neighboring relationships between areas, is
$$
Ad=\left[
\begin{array}{cccc}
0 & \,1\, & \,0 \, & \,0\, \\
1 & 0 & 1 & 0\\
0 & 1 & 0 & 1\\
0 & 0 & 1 & 0
\end{array}
\right]
$$
namely, $Ad_{ij} \neq 0$ if and only if $\frac{P_{ij}}{2 H_i} \neq 0$. In Figure \ref{fig:PrimaComponente} we depict $\Delta \theta_1$ and its estimate $\Delta \hat{\theta}_1$ generated by the DKF algorithm.\\
\begin{figure}
\centering
\includegraphics[width=0.5\textwidth]{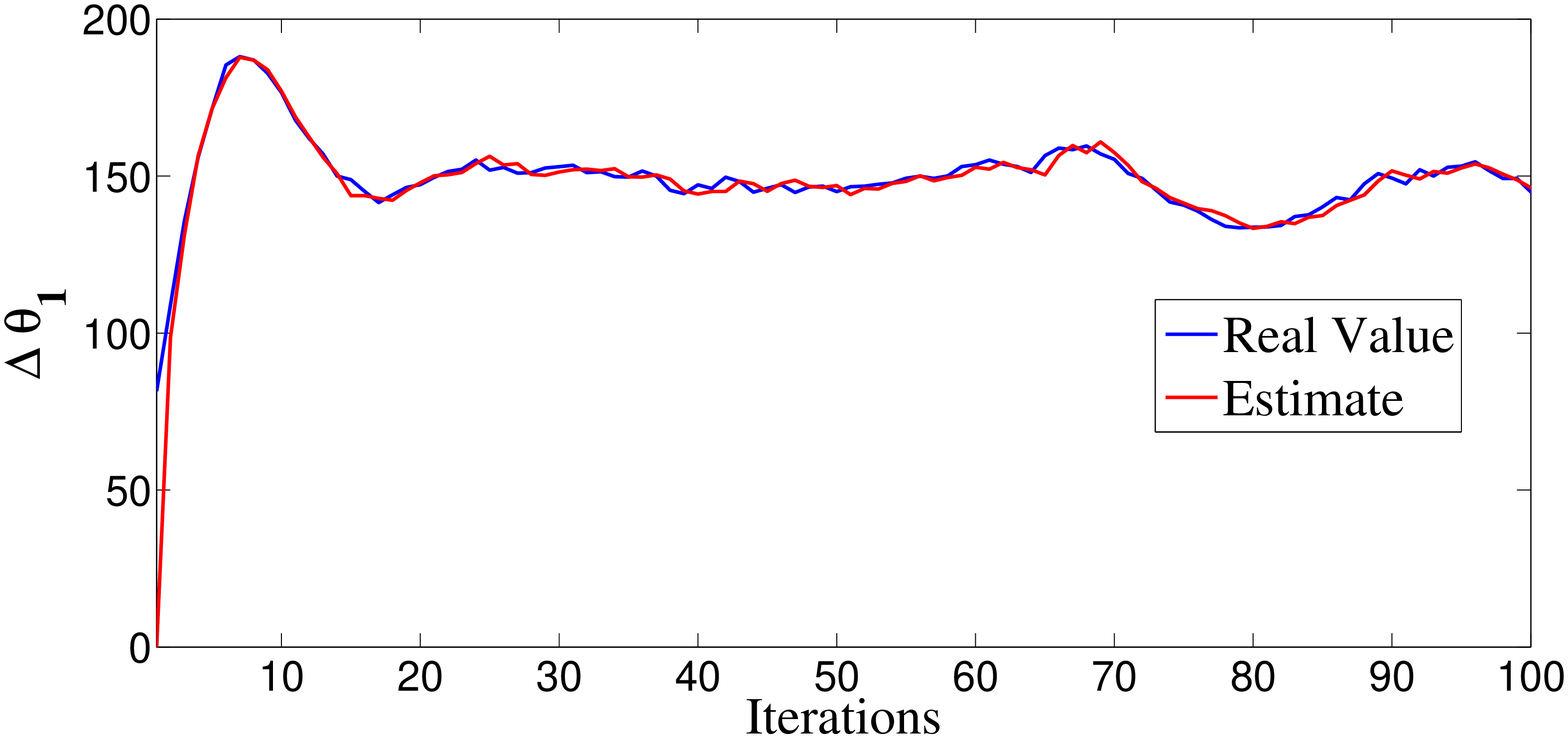}
\caption{Trajectory of $\Delta \theta_1$ (blue line) and its estimate (red line), obtained with DKF}
\label{fig:PrimaComponente}
\end{figure}
In Figures \ref{fig:Comparison} and \ref{fig:Comparison1} we compare the performance of DKF algorithm with that of the centralized Kalman predictor and of the distributed strategy proposed in \cite{Negenborn-Kalman13}. In Figure \ref{fig:Comparison} we plot the normalized estimation error $e(t)$ defined as
$$
e(t)= \frac{1}{\sqrt{M}} \|x(t)- \hat{x}(t) \|
$$
for the first $100$ iterations. In Figure \ref{fig:Comparison1} we plot $e(t)$ from $t=30$ up to $t=100$ (i.e., in stationary conditions).
\begin{figure}
\centering
\includegraphics[width=0.5\textwidth]{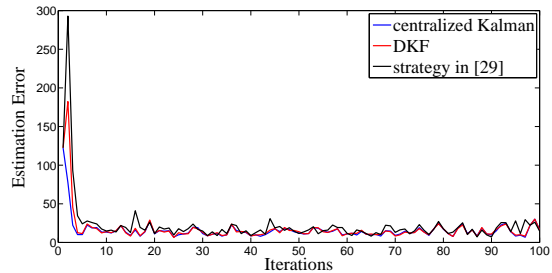}
\caption{Trajectory of $e(t)$ obtained with DKF (red line), with a centralized Kalman predictor (blue line), and with the method proposed in \cite{Negenborn-Kalman13} (black line), with $t\in[0,100]$.}
\label{fig:Comparison}
\end{figure}
\begin{figure}
\centering
\includegraphics[width=0.5\textwidth]{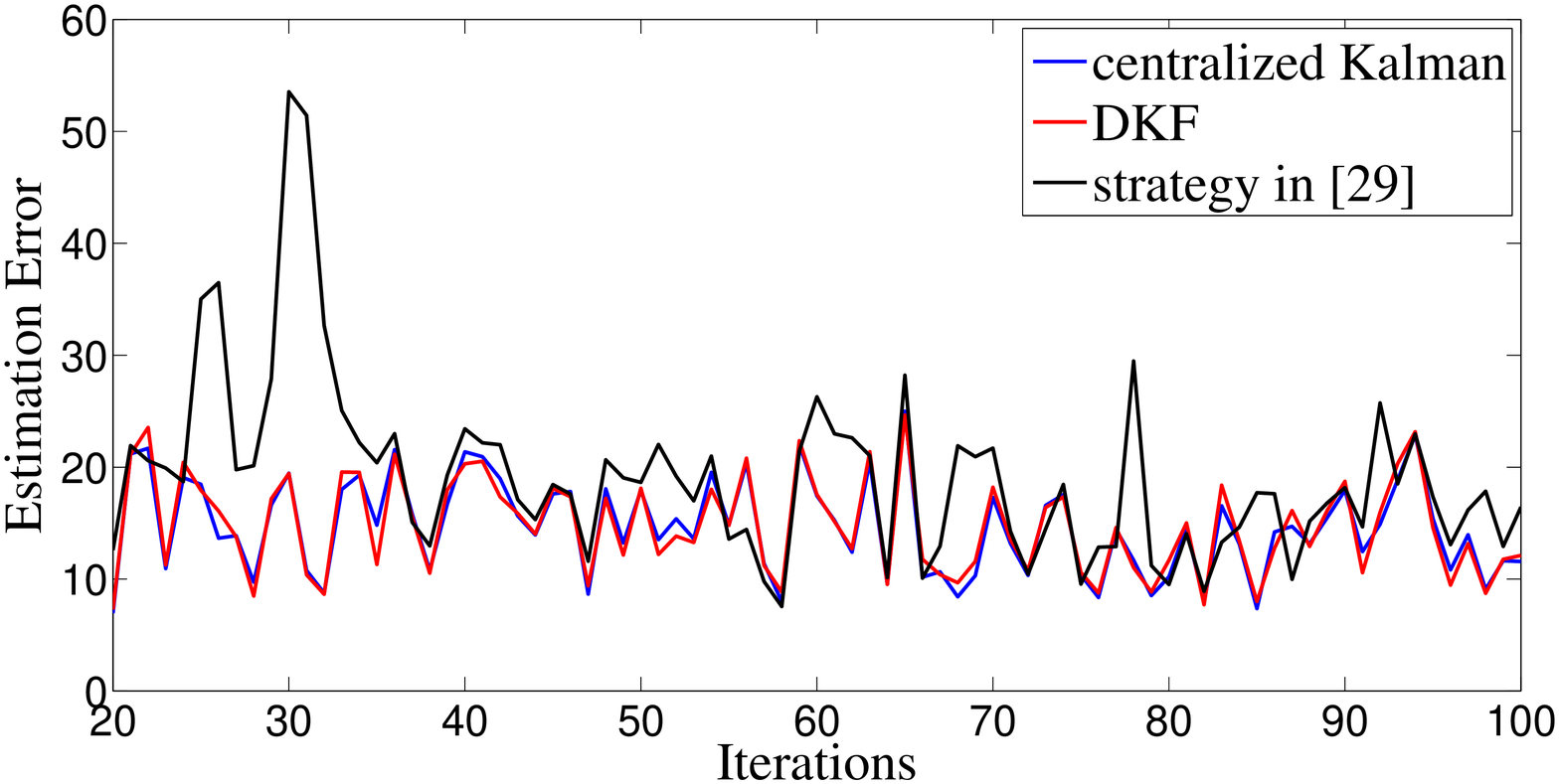}
\caption{Trajectory of $e(t)$ obtained with DKF (red line), with a centralized Kalman predictor (blue line), and with the method proposed in \cite{Negenborn-Kalman13} (black line), with $t\in[20,100]$.}
\label{fig:Comparison1}
\end{figure}
In the Table \ref{table:TabellaMedia} we report the average value of the estimation error evaluated between iteration $30$ and $100$.
\begin{center}
\begin{tabular}{r|c|c|c|}\label{table:TabellaMedia}
&Centralized&\,\,\,DKF\,\,\,& Strategy in [29]\\ \hline
Error Mean&13.74&14.08&17.21 \\ \hline
\end{tabular}
\end{center}
Notice that the performance of DKF algorithm is quite close to the performance of the centralized Kalman filter and that it outperforms the performance of the strategy in \cite{Negenborn-Kalman13}. Additionally Assumption \ref{ass:small-gain-distributed} is satisfied with $H_i=I_4$.

\subsubsection{Plug and play scenario}\label{ex:PlugPlay}\hfill\\
In this section we consider a PnP scenario. Specifically we assume that at time step $50$ a new area (i.e., area $5$) is added to the power network, and that, in particular, it gets connected to area $2$. Again the values of the parameters defining area $5$ can be found in \cite{SR-GFT:12}. The adjacency matrix describing the interconnections after step $50$ is
$$
Ad=\left[
\begin{array}{ccccc}
0 & \,1\, & \,0 \, & \,0 &\, 0\\
1 & 0 & 1 & 0 & 1\\
0 & 1 & 0 & 1 & 0\\
0 & 0 & 1 & 0 & 0\\
0 & 1 & 0 & 0 & 0
\end{array}
\right]
$$
In Figure \ref{fig:PlugPlay} we depict the behavior of the estimation error for both the centralized Kalman filtering algorithm and DKF. Observe that, also in this plug and play scenario the performance of the DKF algorthm is comparable with that of the centralized Kalman filter.
As expected, when a new area is added to the network the value of $e(t)$ increases mainly due to the poor estimation quality concerning the state of the area $5$. However, after few iterations the value of $e(t)$ settles around a value which is comparable to its value before the addition of the new area.

\begin{figure}
\centering
\includegraphics[width=0.5\textwidth]{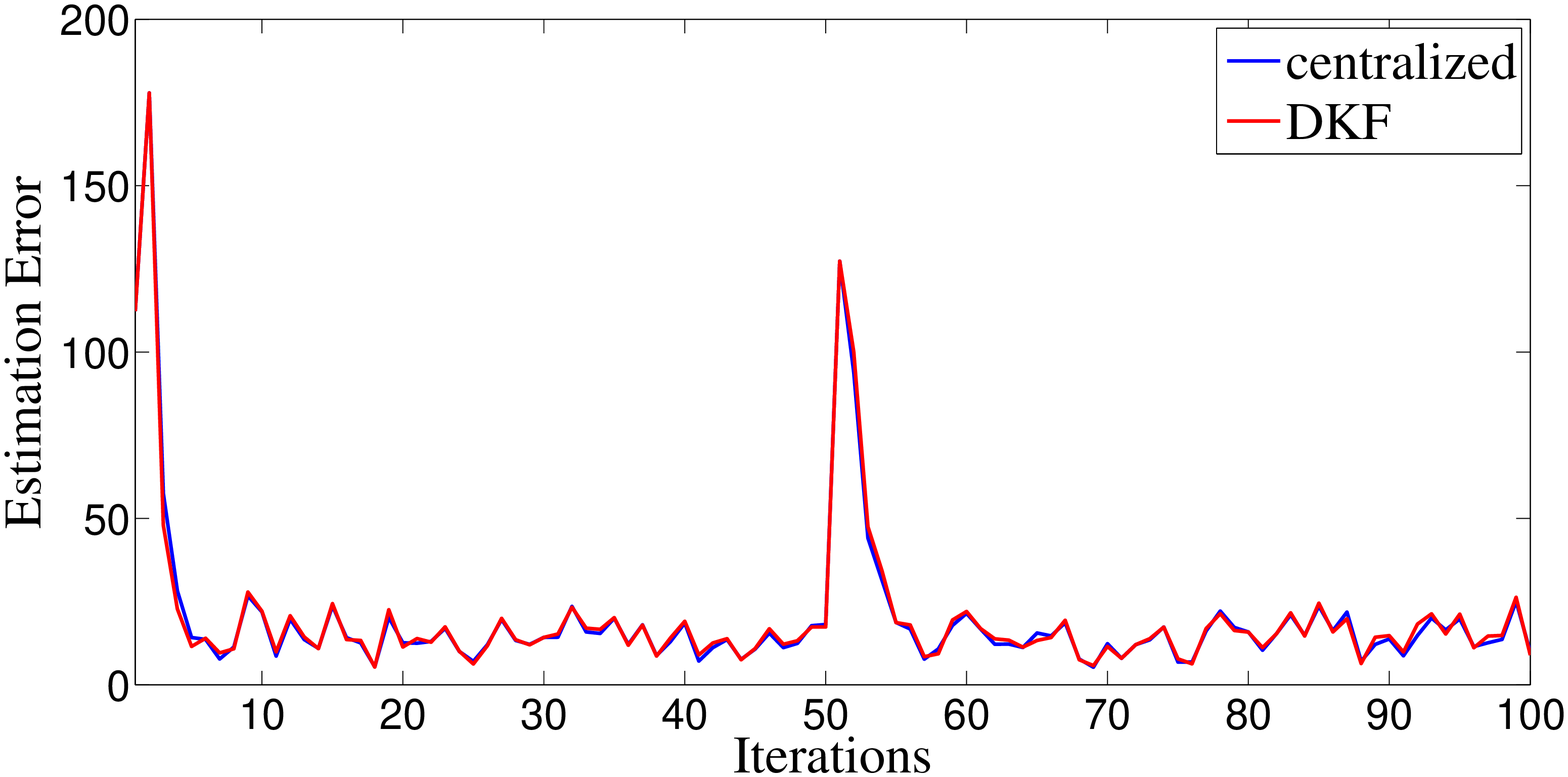}
\label{fig:PlugPlay}
\caption{Trajectory of $e(t)$ obtained with DKF (red line) and with a centralized Kalman predictor (blue line) in the PnP scenario.}
\end{figure}
\section{Conclusions}
\label{sec:conclusions}
In this paper a novel partition-based distributed observer based Kalman filter, named DKF, is proposed. The main advantages of the discussed state estimator are: (\emph{i}) scalability, in terms of both computational and communication loads required for the online operations; (\emph{ii}) the convergence properties can be proved under mild conditions; (\emph{iii}) distributed and plug and play design are allowed. In fact, not only centralized (although aggregate) but also distributed conditions for estimation convergence are given, which confer reconfigurability to the proposed estimation scheme. Simulation tests are provided to illustrate the effectiveness of DKF. For example, we have considered a well-known benchmark example, proposed in the framework of the Hycon2 Project.
Future work include the application of DKF to a real test case, e.g. smart grids.
\appendix
\section{Proofs}
\label{app:proofs}
The following preliminary result is needed for the proofs of both Lemma \ref{lemma:stability} and Proposition \ref{prop:stability}.

\begin{lemma}
Define $\mathbf{P}(k)=\mathrm{diag}(P_1(k),\dots,P_M(k))$. If $P_i(k)$ $i=1,\dots,M$ are updated according to~\eqref{eq:riccati_distr01}, then
\begin{align}
\mathbf{P}(k+1)\geq\mathcal{R}(\mathbf{P}(k),\mathbf{A},\mathbf{C},\mathbf{Q},\mathbf{R})\label{eq:riccati_ineq00}.
\end{align}
\label{prop:distr}
\end{lemma}
\emph{Proof of Lemma~\ref{prop:distr}}\\
Since $\mathbf{P}(k)$ is block-diagonal define
$$\begin{array}{lcl}\mathbf{P}^F(k)&=&\mathbf{P}(k)-\mathbf{P}(k)\mathbf{C}^T(\mathbf{C}\mathbf{P}(k)\mathbf{C}^T+\mathbf{R})^{-1}\mathbf{C}\mathbf{P}(k)\\
&=& \text{diag}(P^F_1(k),\dots,P^F_M(k))\end{array}$$
where
\begin{equation}\label{eq:PF_def}P^F_i(k)=P_i(k)-P_i(k)C_i^T(C_iP_i(k)C_i^T+R_i)^{-1}C_iP_i(k)\geq 0\end{equation}
for all $i=1,\dots,M$. Furthermore, it holds that
\begin{equation}\mathbf{A}\mathbf{P}^F(k)\mathbf{A}^T+\mathbf{Q}\leq \text{diag}(P_1(k+1),\dots,P_M(k+1))\label{eq:Riccati_ineq}\end{equation}
where
\begin{equation}P_i(k+1)=\sum_{j=1}^M \varsigma_j A_{ij}P_j^F(k)A_{ij}^T + Q_i\label{eq:Riccati_ineq_part2}\end{equation}
which is equivalent to \eqref{eq:riccati_distr01}. Inequality \eqref{eq:Riccati_ineq} can be proved as follows. Define a vector $\mathbf{v}=(v_1,\dots,v_M)$, where $v_{i}\in\mathbb{R}^{n_i}$ for all $i=1,\dots,M$. We compute that $\mathbf{v}^T\mathbf{A}\mathbf{P}^F(k)\mathbf{A}^T\mathbf{v}=$
\begin{equation}\begin{array}{lcl}&=&\begin{bmatrix}\sum_{i=1}^M v_i^TA_{i1}&\dots&\sum_{i=1}^M v_i^TA_{iM}\end{bmatrix}\mathbf{P}^F(k)
\begin{bmatrix}\sum_{i=1}^M A_{i1}^Tv_i\\
\vdots\\
\sum_{i=1}^M A_{iM}^T v_i\end{bmatrix}\\
&=&\sum_{j=1}^M\left(\sum_{i=1}^M w_{ij}^T \sum_{i=1}^Mw_{ij}\right)\end{array}\label{eq:proof_7}\end{equation}
where $w_{ij}=\sqrt{P_j^F(k)}A_{ij}^Tv_i$. Remark that $w_{ij}=0$ identically iff $A_{ij}=0$, and that the number of nonzero vectors $w_{\cdot\, j}$ is equal to $\varsigma_j$. We compute that
$\sum_{i=1}^M w_{ij}^T \sum_{i=1}^Mw_{ij}=\sum_{r,s\in{\mathcal{S}}_j} w_{rj}^Tw_{sj}$. Note that, since $\|w_{sj}-w_{rj}\|^2\geq 0$, $w_{rj}^Tw_{sj}\leq \frac{1}{2}(w_{rj}^Tw_{rj}+w_{sj}^Tw_{sj})$. Therefore
$\sum_{r,s\in{\mathcal{S}}_j} w_{rj}^Tw_{sj}\leq \frac{1}{2}\sum_{r,s\in{\mathcal{S}}_j}(\|w_{rj}\|^2+\|w_{sj}\|^2)=\varsigma_j \sum_{i\in{\mathcal{S}}_j}\|w_{ij}\|^2= \sum_{i\in\mathcal{S}_j}\|v_{i}\|^2_{\varsigma_j A_{ij}P_j^F(k)A_{ij}^T}$.\\
From this, it follows that
$$\begin{array}{ll}&\sum_{j=1}^M\left(\sum_{i=1}^M w_{ij}^T \sum_{i=1}^Mw_{ij}\right)\leq \sum_{j=1}^M \sum_{i=1}^M  \|v_{i}\|^2_{\varsigma_j A_{ij}P_j^F(k)A_{ij}^T}\\
&=
  \sum_{i=1}^M  \|v_{i}\|^2_{\sum_{j=1}^M \varsigma_j A_{ij}P_j^F(k)A_{ij}^T}\\
  &=\mathbf{v}^T \text{diag}
  (\sum_{j=1}^M \varsigma_j A_{1j}P_j^F(k)A_{1j}^T,\dots,
  \sum_{j=1}^M \varsigma_j A_{Mj}P_j^F(k)A_{Mj}^T) \mathbf{v}\end{array}$$
  from which \eqref{eq:Riccati_ineq} readily follows.\hfill{}$\square$
\medskip
\emph{Proof of Lemma~\ref{lemma:stability}}\\
From Lemma~\ref{prop:distr}, one has that
\begin{align}
\bar{\mathbf{P}}\geq (\mathbf{A}-\bar{\mathbf{L}}\mathbf{C})\bar{\mathbf{P}}(\mathbf{A}-\bar{\mathbf{L}}\mathbf{C})^T+\mathbf{Q}+\bar{\mathbf{L}}\mathbf{R}\bar{\mathbf{L}}^T
\label{eq:riccati_ineq_01}
\end{align}
where the block entries of $\bar{\mathbf{L}}$ are $\bar{L}_{ij}=\mathcal{L}(\bar{P}_j,A_{ij},C_j,R_j)$, which is equivalent to $\bar{\mathbf{L}}=\mathcal{L}(\bar{\mathbf{P}},\mathbf{A},\mathbf{C},\mathbf{R})$. The latter follows from the fact that $\bar{\mathbf{L}}=\mathbf{A}\text{diag}(L_1^F,\dots,L_M^F)$, where
$\text{diag}(L_1^F,\dots,L_M^F)=\bar{\mathbf{P}}\mathbf{C}^T(\mathbf{C}\bar{\mathbf{P}}\mathbf{C}^T+\mathbf{R})^{-1}$,
which is block-diagonal in view of the block-diagonality of $\mathbf{C}$, $\bar{\mathbf{P}}$, and $\mathbf{R}$.
Assume, by contradiction, that $(\mathbf{A}-\bar{\mathbf{L}}\mathbf{C})$ is not Schur stable. Therefore, there is at least an eigenvalue/eigenvector pair $\lambda,v$ of $(\mathbf{A}-\bar{\mathbf{L}}\mathbf{C})$ such that $(\mathbf{A}-\bar{\mathbf{L}}\mathbf{C})^Tv=\lambda v$ and $|\lambda|\geq 1$. From \eqref{eq:riccati_ineq_01}
$$v^T\bar{\mathbf{P}}v \geq v^T(\mathbf{A}-\bar{\mathbf{L}}\mathbf{C})\bar{\mathbf{P}}(\mathbf{A}-\bar{\mathbf{L}}\mathbf{C})^T v+v^T\mathbf{Q}v+v^T\bar{\mathbf{L}}\mathbf{R}\bar{\mathbf{L}}^Tv
$$
from which it follows that
$(1-|\lambda|^2) v^T\bar{\mathbf{P}}v \geq v^T\mathbf{Q}v+v^T\bar{\mathbf{L}}\mathbf{R}\bar{\mathbf{L}}^Tv$.
Since the right hand side of the latter inequality is $\geq 0$ and $|\lambda|\geq 1$, the only possibility is that $|\lambda|=1$, $v^T\mathbf{Q}v=0$, and $\bar{\mathbf{L}}^Tv=0$. In view of this, $\mathbf{A}^T v=v$ and $\mathbf{G}^T v$ should hold at the same time which, recalling the PBH test, is in contradiction with the assumption that the pair $(\mathbf{A,G})$ is stabilizable. This concludes the proof of Lemma~\ref{lemma:stability}.\hfill$\square$

\medskip
\emph{Proof of Proposition \ref{prop:stability}}

\medskip

As a preliminary step, we show that, if $\Pi_d(1)\leq \mathbf{P}(1)$, then  $\Pi_d(k)\leq \mathbf{P}(k)$ for all $k\geq 0$. This can be proved using induction arguments. Assume that, at instant $k$, it holds that $\Pi_d(k)\leq \mathbf{P}(k)$. Recalling  Lemma \ref{prop:distr}, we have that
$\mathbf{P}(k+1)\geq \mathcal{R}(\mathbf{P}(k),\mathbf{A},\mathbf{C},\mathbf{Q},\mathbf{R})$,
where $\mathcal{R}(\mathbf{P}(k),\mathbf{A},\mathbf{C},\mathbf{Q},\mathbf{R})=(\mathbf{A}-\mathbf{L}(k)\mathbf{C})\mathbf{P}(k)(\mathbf{A}-\mathbf{L}(k)\mathbf{C})^T+\mathbf{L}(k)\mathbf{R}\mathbf{L}(k)^T+\mathbf{Q}$.
From this and \eqref{eq:prediction_error_collective_var} it results that
$$\mathbf{P}(k+1)-\Pi_d(k+1)\geq (\mathbf{A}-\mathbf{L}(k)\mathbf{C})(\mathbf{P}(k)-\Pi_d(k))(\mathbf{A}-\mathbf{L}(k)\mathbf{C})^T\geq 0$$
Therefore, $\Pi_d(k+1)\leq \mathbf{P}(k+1)$. By applying induction arguments, we can prove that $\Pi_d(k)\leq \mathbf{P}(k)$ for all $k\geq 0$.\\
If ${\mathbf{P}}(k)\rightarrow \bar{\mathbf{P}}$ as $k\rightarrow\infty$, then $\mathbf{L}(k)\rightarrow \bar{\mathbf{L}}$ such that, in view of Lemma \ref{lemma:stability}, $\mathbf{A}-\bar{\mathbf{L}}\mathbf{C}$ is Schur stable. Consider the evolution of matrix $\Pi_d(k)$. From the stability of $\mathbf{A}-\bar{\mathbf{L}}\mathbf{C}$, then $\Pi_d(k)\rightarrow \bar{\Pi}_d$, for all initial conditions $\Pi_d(1)$, where $\bar{\Pi}_d$ is the unique solution to the Lyapunov equation
$\bar{\Pi}_d=(\mathbf{A}-\bar{\mathbf{L}}\mathbf{C})\bar{\Pi}_d
(\mathbf{A}-\bar{\mathbf{L}}\mathbf{C})^T+\mathbf{Q}+
\bar{\mathbf{L}}\mathbf{R}\bar{\mathbf{L}}^T$.
If we set $\Pi_d(1)=0$, from the preliminary result then $\Pi_d(k)\leq \mathbf{P}(k)$ for all $k\geq 0$ and $\bar{\Pi}_d\leq \bar{\mathbf{P}}$. Noting that $\bar{\Pi}_d$ is the unique steady-state attained for all possible initial conditions the proof is concluded. \hfill$\square$

\medskip
The proof of Theorem \ref{prop:small-gain} heavily relies on classical results on Kalman filters, e.g., \cite{Caines-Mayne-Riccati,AndMooFilt,GoodwinRiccati-84,GoodwinRiccati-86}.
Similarly to well known results on the discrete-time Riccati equation, we need two intermediate results.

\begin{lemma}
\label{lemma:riccati:monoton}
If $P_j^A(k)\geq P_j^B(k)$ for all $j=1,\dots,M$, then $P_i^A(k+1)\geq P_i^B(k+1)$ where $P_i^A(k+1)$ and $P_i^B(k+1)$ are the matrix evolutions, obtained with \eqref{eq:riccati_distr01}, starting from $P_j^A(k)$ and $P_j^B(k)$, respectively.
\end{lemma}
\begin{proof}
%
Note that we can write \eqref{eq:riccati_distr01} as
\begin{equation}
\label{eq:riccati_distr03}
\begin{array}{ll}
P_i(k+1)=&\sum_{j=1}^M (\tilde{A}_{ij}-{L}_{ij}(k)\tilde{C}_j)P_j(k)(\tilde{A}_{ij}-{L}_{ij}(k)\tilde{C}_j)^T\\
&+ {L}_{ij}(k)\tilde{R}_j{L}_{ij}(k)^T + Q_i\end{array}
\end{equation}
where, according to the classical Kalman filter theory, ${L}_{ij}(k)=\mathcal{L}(P_j(k),\tilde{A}_{ij},\tilde{C}_j,\tilde{R}_j)$ minimizes the term $(\tilde{A}_{ij}-{L}_{ij}(k)\tilde{C}_j)P_j(k)(\tilde{A}_{ij}-{L}_{ij}(k)\tilde{C}_j)^T+ {L}_{ij}(k)\tilde{R}_j{L}_{ij}(k)^T$ for all $i,j=1,\dots,M$. Therefore, consider the matrices $P_i^A$, $P_i^B$, where $P_i^A\geq P_i^B$ for all $i=1,\dots,M$, and optimal the gains ${L}_{ij}^A$ and ${L}_{ij}^B$ corresponding to $P_i^A$ and $P_i^B$, respectively, then for all $j$, $(\tilde{A}_{ij}-{L}_{ij}^B\tilde{C}_j)P_j^B(\tilde{A}_{ij}-{L}_{ij}^B
\tilde{C}_j)^T+ {L}_{ij}^B\tilde{R}_j({L}_{ij}^B)^T \leq (\tilde{A}_{ij}-{L}_{ij}^A\tilde{C}_j)P_j^B(\tilde{A}_{ij}-{L}_{ij}^A
\tilde{C}_j)^T+ {L}_{ij}^A\tilde{R}_j({L}_{ij}^A)^T \leq (\tilde{A}_{ij}-{L}_{ij}^A\tilde{C}_j)P_j^A(\tilde{A}_{ij}-{L}_{ij}^A
\tilde{C}_j)^T+ {L}_{ij}^A\tilde{R}_j({L}_{ij}^A)^T$,
and the proof is concluded.\hfill$\square$
\end{proof}

\begin{lemma}
\label{lemma:riccati:boundedness_small-gain}
If Assumptions \ref{ass:local-in} (for all $i=1,\dots,M$) and \ref{ass:small-gain-global} hold then, for all $P_i(1)\geq 0$ $i=1,\dots,M$,
there exist $P_i^{MAX}$ for all $i=1,\dots,M$ such that $P_i(k)\leq P_i^{MAX}$ for all $k\geq 0$.
\end{lemma}
\begin{proof}
%
An alternative formulation of \eqref{eq:riccati_distr01} is
\begin{equation}
P_i(k)=P_i^L(k)+\Delta_i(k)\label{eq:riccati_distr02}
\end{equation}
where $P_i^L(k+1)
=\tilde{A}_{ii} P^F_i(k) \tilde{A}_{ii}^T+Q_i
=(\tilde{A}_{ii}-{L}_{ii}(k)\tilde{C}_i)P_i(k)(\tilde{A}_{ii}-{L}_{ii}(k)\tilde{C}_i)^T+{L}_{ii}(k)\tilde{R}_i{L}_{ii}^T(k)+Q_i
$ and $\Delta_i(k+1)=\sum_{j\neq i}\tilde{A}_{ij} P^F_j(k) \tilde{A}_{ij}^T$, being $P^F_j(k)$ defined in \eqref{eq:PF_def}.
In view of Assumption \ref{ass:local-in}, we can write $P^F_i(k-1)=\tilde{A}_{ii}^{-1}(P_i^L(k)-Q_i)(\tilde{A}_{ii}^{-1})^T$. Therefore
\begin{equation}
\label{eq:Delta_def}\Delta_i(k)=\sum_{j\neq i}\tilde{A}_{ij}\tilde{A}_{jj}^{-1}({P}_j^L(k)-Q_j)(\tilde{A}_{jj}^{-1})^T\tilde{A}_{ij}^T\end{equation}
Since $\bar{F}_i=(\tilde{A}_{ii}-\bar{L}_i \tilde{C}_i)$ is Schur stable (thanks to Assumption~\ref{ass:small-gain-global} (i)) and $\bar{L}_i$ is a suboptimal gain
$$P_{L}(k+1)\leq \bar{F}_i
(P_i^L(k)+\Delta_i(k))
\bar{F}_i^T+Q_i+\bar{L}_i\tilde{R}_i(\bar{L}_i)^T$$
Solving the latter we obtain:
\begin{align}
&P_i^L(k)\leq \bar{F}_i^{k-1} P_i^L(1) (\bar{F}_i^T)^{k-1}\label{eq:deltaP}\\
&+ \sum_{h=1}^{k-1}\bar{F}_i^{h-1}\left( \bar{F}_i\Delta_i(k-h)\bar{F}_i^T+ Q_i+\bar{L}_i\tilde{R}_i(\bar{L}_i)^T\right)(\bar{F}_i^T)^{h-1}\nonumber
\end{align}
Using the transformation matrices $H_i$, we define
$\hat{P}_i^L(k)=H_i P_i^L(k)H_i^T$, $\hat{Q}_i=H_i Q_i H_i^T$ $\hat{\hat{Q}}_i=H_i (Q_i+\bar{L}_i\tilde{R}_i(\bar{L}_i)^T) H_i^T$.
Note also that $H_i\tilde{A}_{ij}\tilde{A}_{jj}^{-1}H_j^{-1}=
H_i{A}_{ij}H_j^{-1}H_j{A}_{jj}^{-1}H_j^{-1}=\hat{A}_{ij}\hat{A}_{jj}^{-1}$.
In view of this and \eqref{eq:Delta_def}, we can rewrite \eqref{eq:deltaP} as follows
\begin{align}
&\hat{P}_i^L(k)\leq \hat{F}_i^{k-1} \hat{P}_i^L(1) (\hat{F}_i^T)^{k-1} + \sum_{h=1}^{k-1}\hat{F}_i^{h-1}\hat{\hat{Q}}_i(\hat{F}_i^T)^{h-1}\nonumber\\
&+\sum_{j\neq i} \sum_{h=1}^{k-1}\hat{F}_i^{h}
\hat{A}_{ij}\hat{A}_{jj}^{-1}(\hat{P}_j^L(k-h)-\hat{Q}_j)(\hat{A}_{jj}^{-1})^T\hat{A}_{ij}^T
(\hat{F}_i^T)^{h}
\label{eq:deltaPhat}
\end{align}
Recalling that ${P}_j^L(k)\geq {P}_j^L(k)-Q_j\geq 0$ for all $j=1,\dots,M$, we have that
\begin{align}
\|\hat{P}_i^L(k)\|\leq \|\hat{F}_i^{k-1}\|^2 \|\hat{P}_i^L(1)\|+\sum_{h=1}^{k-1}\|\hat{F}_i^h\|^2\|\hat{\hat{Q}}_i\|\nonumber\\
+
\sum_{j\neq i} \|\hat{A}_{ij}\hat{A}_{jj}^{-1}\|^2 \sum_{h=1}^{k-1}\|\hat{F}_i^h\|^2\|\hat{P}_j^L(k-h)\|
\end{align}
Therefore
\begin{align}\label{eq:pseudoISS}
\|\hat{P}_i^L(k)\|\leq \mu_i^2 \lambda_i^{2k}\|\hat{P}_i^L(1)\|+\frac{\mu_i^2}{1-\lambda_i^2}\|\hat{\hat{Q}}_i\|\nonumber\\
+ \sum_{j\neq i} \|\hat{A}_{ij}\hat{A}_{jj}^{-1}\|^2(\max_{h\in[0,k]}\|\hat{P}_j^L(h)\|) \frac{\mu_i^2}{1-\lambda_i^2}
\end{align}
Denoting $n_j(k) = \max_{h\in[0,k]}\|\hat{P}_j^L(h)\|$, \eqref{eq:pseudoISS} implies that
$0\leq n_i(k) \leq q_i + \sum_{j\neq i}\gamma_{ij} n_j(k)$
where $q_i=\mu_i^2 \|\hat{P}_i^L(1)\|+\frac{\mu_i^2}{1-\lambda_i^2}\|\hat{\hat{Q}}_i\|$ and $\gamma_{ij}=\frac{\mu_i^2}{1-\lambda_i^2}\|\hat{A}_{ij}\hat{A}_{jj}^{-1}\|^2$. Finally denote the vectors $\mathbf{n}(k)=(n_1(k),\dots,n_M(k))$ and $\mathbf{q}=(q_1,\dots,q_M)$. We obtain that
\begin{align}
(I_{M}-\Gamma)\mathbf{n}(k)\leq \mathbf{q}
\label{eq:nbold}
\end{align}
According to \cite{Wirth-Small-gain07}, if the spectral radius of $\Gamma$ is strictly smaller than one, for every initial condition (see, e.g., Lemma 13 for the general nonlinear case), the solution to the system \eqref{eq:riccati_distr02} exists and is uniformly bounded, since $\mathbf{q}$ does not depend on $k$.\hfill{}$\square$\\
\end{proof}

\medskip
Now we are in the position to provide the proof of Theorem \ref{prop:small-gain}.

\medskip

\emph{Proof of Theorem \ref{prop:small-gain}}

\medskip
First consider all the initializations a, b, and c.

\medskip
a. In case $P_i(1)=0$ for all $i=1,\dots,M$ then $P_i(2)\geq 0=P_i(1)$ for all $i=1,\dots,M$.

\medskip
b. Set $P_i(1)=\bar{P}_i^N$ for all $i=1,\dots,M$.
Note that matrices $\bar{P}_i^N$ exist and are unique for all $i=1,\dots,M$ in view of Assumption \ref{ass:local-det1}.
From \eqref{eq:riccati_distr01}, for all $i=1,\dots,M$
$$\begin{array}{lcl}
P_i(2)&=&\varsigma_i\mathcal{R}(\bar{P}_i^N,{A}_{ii},{C}_i,{Q}_i,{R}_i)+
\sum_{j\in\mathcal{N}_i\setminus\{i\}}
\varsigma_j\mathcal{R}(\bar{P}_j^N,{A}_{ij},{C}_j,0,{R}_j)\nonumber\\
&\geq&\mathcal{R}(\bar{P}_i^N,{A}_{ii},{C}_i,{Q}_i,{R}_i)=\bar{P}_i^N=P_i(1)
\end{array}$$

\medskip
c. Set $P_i(1)=\tilde{P}_i^N$ for all $i=1,\dots,M$.
Note that matrices $\tilde{P}_i^N$ exist and are unique (for all $i=1,\dots,M$) in view of Assumption \ref{ass:local-det2}.
Set $P_i(1)=\tilde{P}_i^N$ for all $i=1,\dots,M$. Then, from \eqref{eq:riccati_distr01}, for all $i=1,\dots,M$
$$\begin{array}{lcl}
P_i(2)&=&\mathcal{R}(\tilde{P}_i^N,\tilde{A}_{ii},\tilde{C}_i,{Q}_i,\tilde{R}_i)+
\sum_{j\in\mathcal{N}_i\setminus\{i\}}
\mathcal{R}(\tilde{P}_j^N,\tilde{A}_{ij},\tilde{C}_j,0,\tilde{R}_j)\nonumber\\
&\geq&\mathcal{R}(\tilde{P}_i^N,\tilde{A}_{ii},\tilde{C}_i,{Q}_i,\tilde{R}_i)=\tilde{P}_i^N=P_i(1)
\end{array}$$

\medskip
In all cases, applying induction arguments and in view of the monotonicity property (i.e., Lemma \ref{lemma:riccati:monoton}), $P_i(k+1)\geq P_i(k)$ for all $k\geq 1$ and for all $i=1,\dots,M$. Therefore the sequence of matrices $\mathbf{P}(k)=$diag$(P_1(k),\dots,P_M(k))$ is monotonically increasing, in the sense that $\mathbf{P}(k+1)\geq \mathbf{P}(k)$ for all $k$. In view of the boundedness property  (i.e., Lemma \ref{lemma:riccati:boundedness_small-gain}), there exist $\bar{P}_i$ for all $i$, such that $P_i(k)\rightarrow \bar{P}_i$ as $k\rightarrow\infty$.\hfill{}$\square$

\medskip
\emph{Proof of Proposition \ref{propo:Gamma-distributed}}\\
The proof easily follows from the Gershgorin circle theorem. Indeed, each eigenvalue of $\Gamma$ lies in at least one of the $M$ Gershgorin circles, i.e., since $\gamma_{ii}=0$ for all $i$, the values of $\lambda$ satisfying $|\lambda|\leq \rho_i=\sum_{j=1}^M |\gamma_{ij}|=\sum_{j=1}^M\gamma_{ij}$, for each $i=1,\dots,M$. Then, if $\rho_i<1$ for all $i=1,\dots,M$, all eigenvalues verify $|\lambda|<1$.\hfill$\square$

\medskip
\emph{Proof of Corollary \ref{cor:small-gain-Plugin}}\\
When plug-in events take place, if $j\leq M$ is a neighbor (also said \emph{predecessor} in \cite{RiversoFarinaGFT_PnP13}) of $M+1$, then $\varsigma_j^+=\varsigma_j+1$, otherwise $\varsigma_j^+=\varsigma_j$. In view of this, $\tilde{A}^+_{ij}=\sqrt{\frac{\varsigma_j+1}{\varsigma_j}}\tilde{A}_{ij}$ (for $i\in\mathcal{S}_{j}$), $\tilde{C}^+_{j}=\sqrt{\frac{\varsigma_j+1}{\varsigma_j}}\tilde{C}_{j}$, and $\tilde{R}^+_j=\frac{\varsigma_j+1}{\varsigma_j}\tilde{R}_j$ for all $j\in\mathcal{N}_{M+1}$; otherwise $\tilde{A}^+_{ij}=\tilde{A}_{ij}$, $\tilde{C}^+_{j}=\tilde{C}_{j}$, and $\tilde{R}^+_j=\tilde{R}_j$. Therefore, for all initializations and for all $i=1,\dots,M$
$$\begin{array}{lcl}
P_i(T_{PnP}+1)&=&\sum_{j\in\mathcal{N}_i^+}
\mathcal{R}(P_j(T_{PnP}),\tilde{A}_{ij}^+,\tilde{C}_j^+,0,\tilde{R}_j^+)+Q_i\\
&=&
\sum_{j\in\mathcal{N}_i}\frac{\varsigma_j^+}{\varsigma_j}
\mathcal{R}(P_j(T_{PnP}),\tilde{A}_{ij},\tilde{C}_j,0,\tilde{R}_j)+Q_i\\
&&+\mathcal{R}(P_{M+1}(T_{PnP}),\tilde{A}_{i(M+1)}^+,\tilde{C}_{M+1}^+,0,\tilde{R}_{M+1}^+)\\
&\geq& \sum_{j\in\mathcal{N}_i}
\mathcal{R}(\bar{P}_j,\tilde{A}_{ij},\tilde{C}_j,0,\tilde{R}_j)+Q_i=\bar{P}_i=P_i(T_{PnP})
\end{array}$$
\begin{itemize}
\item[a.] if $P_{M+1}(T_{PnP})=0$, $P_{M+1}(T_{PnP}+1)\geq 0$.
\item[b.] if $P_{M+1}(T_{PnP})=\bar{P}_{M+1}^N$,
    $P_{M+1}(T_{PnP}+1)\geq\mathcal{R}(\bar{P}_{M+1}^N,{A}_{(M+1)(M+1)},$ ${C}_{M+1},{Q}_{M+1},{R}_{M+1})=\bar{P}_{M+1}^N=P_{M+1}(T_{PnP})$,
\item[c.] if $P_{M+1}(T_{PnP})=\tilde{P}_{M+1}^N$,
    $P_{M+1}(T_{PnP}+1)\geq\mathcal{R}(\bar{P}_{M+1}^N,\tilde{A}_{(M+1)(M+1)}^+,$ $\tilde{C}_{M+1}^+,{Q}_{M+1},\tilde{R}_{M+1}^+)=\tilde{P}_{M+1}^N=P_{M+1}(T_{PnP})$
\end{itemize}
In all cases it follows that, for all $i=1,\dots,M+1$, $P_i(T_{PnP}+1)\geq P_i(T_{PnP})$. Applying an induction argument and in view of the monotonicity property, $P_i(k+1)\geq P_i(k)$ for all $k\geq T_{PnP}$ and for all $i=1,\dots,M$. Therefore the sequence of matrices $\mathbf{P}(k)=$diag$(P_1(k),\dots,P_M(k))$ is monotonically increasing, in the sense that $\mathbf{P}(k+1)\geq \mathbf{P}(k)$ for all $k\geq T_{PnP}$. Since Assumption~\ref{ass:small-gain-global} (ii) holds in view of Assumption~\ref{ass:small-gain-distributed} and Proposition~\ref{prop:distr}, the boundedness property holds in view of Lemma \ref{lemma:riccati:boundedness_small-gain}, and therefore there exist $\bar{P}_i^+$ for all $i=1,\dots,M+1$, such that $P_i(k)\rightarrow \bar{P}_i^+$ as $k\rightarrow\infty$.\hfill$\square$

\medskip
\emph{Proof of Corollary \ref{cor:small-gain-unPlug}}\\
When the unplug event takes place, if $j< M$ is a neighbor of $M$, then $\varsigma_j^+=\varsigma_j-1$, otherwise $\varsigma_j^+=\varsigma_j$. In view of this, $\tilde{A}^+_{ij}=\sqrt{\frac{\varsigma_j-1}{\varsigma_j}}\tilde{A}_{ij}$ if $i\in\mathcal{S}_{j}$, and $\tilde{C}^+_{j}=\sqrt{\frac{\varsigma_j-1}{\varsigma_j}}\tilde{C}_{j}$ and $\tilde{R}^+_j=\frac{\varsigma_j-1}{\varsigma_j}\tilde{R}_j$ for all $j\in\mathcal{N}_{M}$; otherwise $\tilde{A}^+_{ij}=\tilde{A}_{ij}$, $\tilde{C}^+_{j}=\tilde{C}_{j}$, and $\tilde{R}^+_j=\tilde{R}_j$.
Therefore, for all $i=1,\dots,M-1$
$$\begin{array}{lcl}
P_i(T_{PnP}+1)&=&\sum_{j\in\mathcal{N}_i^+}
\mathcal{R}(P_j(T_{PnP}),\tilde{A}_{ij}^+,\tilde{C}_j^+,0,\tilde{R}_j^+)+Q_i\\
&=&
\sum_{j\in\mathcal{N}_i^+}\frac{\varsigma_j^+}{\varsigma_j}
\mathcal{R}(P_j(T_{PnP}),\tilde{A}_{ij},\tilde{C}_j,0,\tilde{R}_j)+Q_i\\
&\leq &
\sum_{j\in\mathcal{N}_i^+}
\mathcal{R}(P_j(T_{PnP}),\tilde{A}_{ij},\tilde{C}_j,0,\tilde{R}_j)+Q_i\\
&= &
\sum_{j\in\mathcal{N}_i}
\mathcal{R}(P_j(T_{PnP}),\tilde{A}_{ij},\tilde{C}_j,0,\tilde{R}_j)+Q_i\\
&&-\mathcal{R}(P_M(T_{PnP}),\tilde{A}_{iM},\tilde{C}_M,0,\tilde{R}_M)\leq \bar{P}_i=P_i(T_{PnP})
\end{array}$$
Then, applying an induction argument and in view of the monotonicity property, $P_i(k+1)\leq P_i(k)$ for all $k\geq T_{PnP}$ and for all $i=1,\dots,M$. Therefore the sequence of matrices $\mathbf{P}(k)=$diag$(P_1(k),\dots,P_M(k))$ is monotonically decreasing, in the sense that $\mathbf{P}(k+1)\leq \mathbf{P}(k)$ for all $k\geq T_{PnP}$. In view of the fact that $\mathbf{P}(k)\geq 0$, there exist $\bar{P}_i^+$ for all $i=1,\dots,M+1$, such that $P_i(k)\rightarrow \bar{P}_i^+$ as $k\rightarrow\infty$.\hfill$\square$

\bibliographystyle{plain}        
\bibliography{DMHE_bib}          

\end{document}